\documentclass[10pt]{article}
\usepackage{dcolumn}
\usepackage{bm}
\usepackage{verbatim}       

\usepackage{amsthm}
\usepackage{amssymb}

\usepackage{booktabs}

\usepackage{amstext}
\usepackage{amsmath}
\usepackage{amsfonts}
\usepackage{units}
\usepackage{mathrsfs}
\usepackage{cite}

\newcommand{\p}{\partial}

\newcommand{\dd}{{\rm d}}

\newcommand{\bd}{\begin{definition}}                
\newcommand{\ed}{\end{definition}}                  
\newcommand{\bc}{\begin{corollary}}                 
\newcommand{\ec}{\end{corollary}}                   
\newcommand{\bl}{\begin{lemma}}                     
\newcommand{\el}{\end{lemma}}                       
\newcommand{\bp}{\begin{proposition}}            
\newcommand{\ep}{\end{proposition}}                
\newcommand{\bere}{\begin{remark}}                  
\newcommand{\ere}{\end{remark}}                     

\newcommand{\bt}{\begin{theorem}}
\newcommand{\et}{\end{theorem}}

\newcommand{\be}{\begin{equation}}
\newcommand{\ee}{\end{equation}}

\newcommand{\bit}{\begin{itemize}}
\newcommand{\eit}{\end{itemize}}
\newtheorem{theorem}{Theorem}[section]
\newtheorem{corollary}[theorem]{Corollary}
\newtheorem{lemma}[theorem]{Lemma}
\newtheorem{proposition}[theorem]{Proposition}
\theoremstyle{definition}
\newtheorem{definition}[theorem]{Definition}
\theoremstyle{remark}
\newtheorem{remark}[theorem]{Remark}

\begin{document}


\title{The connections of pseudo-Finsler spaces}


\author{E. Minguzzi\thanks{
Dipartimento di Matematica e Informatica ``U. Dini'', Universit\`a
degli Studi di Firenze, Via S. Marta 3,  I-50139 Firenze, Italy.
Corresponding author e-mail: ettore.minguzzi@unifi.it }}

\date{}

\maketitle

\begin{abstract}
\noindent We give an  introduction to (pseudo-)Finsler geometry and its connections. For most results we  provide short and self contained proofs. Our study of the Berwald non-linear connection is framed  into the theory of  connections over general fibered spaces pioneered by Mangiarotti, Modugno and other scholars. The main identities for the linear Finsler connection are presented in the general case, and then specialized to some notable cases like Berwald's, Cartan's or Chern-Rund's. In this way it becomes easy to compare them and see the advantages of one connection over the other. Since we introduce two soldering forms we are able to characterize the notable Finsler connections in terms of their torsion properties.  As an application, the curvature symmetries implied by the compatibility with a metric suggest that in Finslerian generalizations of general relativity the mean Cartan torsion  vanishes. This observation allows us to obtain dynamical equations which imply a satisfactory conservation law.
The work ends with a discussion of yet another Finsler  connection which has some advantages over Cartan's and Chern-Rund's.
\end{abstract}

\begin{flushright}
{\em Dedicated to Marco Modugno}
\end{flushright}

\tableofcontents

\section{Introduction}
We introduce Finsler geometry through the theory of connections on generalized fibered manifolds. The theory of connections was first developed in the limited framework of principal bundles. Each connection on a vector bundle was then seen as  induced from another connection on a principal bundle thanks to a representation of the group into a vector space. This somewhat involved approach, of which the reader can find a good account in Kobayashi and Nomizu's classic book \cite{kobayashi63}, is still the most popular among physicists and other practitioners of gauge theories.

Fortunately, with the discovery of the Fr\"olicher-Nijenhuis bracket and with the understanding of its fundamental role in the theory of connections, it became clear that connections could be introduced directly over every fibered space as splittings of the tangent bundle, with only reference to the structure of the fiber and without necessarily mentioning any related principal bundle. The theory worked very well especially for what concerned the definition of curvature and the Bianchi identities. Finally, the work by Marco Modugno \cite{modugno91} clarified how to introduce the torsion passing through  a soldering form, thus completing the theory. The reader wishing to explore the theory of connections in more detail than done in this work is  referred to \cite{modugno91,michor91,kolar93}.

A special laboratory in which one can test these general techniques is Finsler geometry. This geometry is also somewhat involved because it presents several types of connections, whose appropriate presentation becomes crucial for a correct understanding of the theory. Cartan had already understood how to frame this theory in a satisfactorily way using the concept of vertical bundle; still the role of the non-linear connection was not fully  recognized, and hence it is important to compare the results so far available in Finsler geometry with those results on non-linear connections which can be deduced from the general theory just mentioned.

This work does not present an historical account of Finsler geometry, as it rather reflects the author's personal views on this theory. We apologize if some sources  have  not been fully acknowledged. Similarly, the original content of this work might have not been sufficiently stressed.

Good introductions are provided by the textbooks \cite{antonelli93,bao00,shen01,mo06,szilasi14}. Investigations on Finsler geometry possibly related to this one are \cite{grifone72,abate96,szilasi98}. A nice short introduction to Finsler geometry   much  in the spirit of this work  is \cite{dahl06}.

\section{Geometrical preliminaries}
This section is devoted to the introduction of some known concepts. It main purpose is that of establishing  and clarifying the notation.

\subsection{Vector bundle and local coordinates}

Let $\pi_B\colon E\to B$ be a vector bundle. Let $\{x^i\}$ be coordinates on the base and let $E_x\ni y = y^a e_a$ where $\{e_a\}$ is a base of $E_x$, so that $(x^i,y^a)$ are local coordinates on $E$. Every element  $v\in TE$ reads
\begin{equation} \label{vhg}
v=\dot{x}^i \frac{\p }{\p x^i}+\dot{y}^a \frac{\p }{\p y^a}
\end{equation}
so that
\[
(x^i,y^a, \dot{x}^i, \dot{y}^a)
\]
are local coordinates of $TE$. Observe that $\frac{\p }{\p y^a}=e_a$.

If we make a change of coordinates and a change of base
\begin{align}
\tilde{x}&=\tilde{x}(x) ,  \label{bh1}\\
\tilde{e}_a(x)&=(M^{-1})^b_{\,\, a}(x) \,e_b, \label{alf}
\end{align}
 then
\begin{equation} \label{lol}
\tilde y^a=M^a_{\, b}(x)\, y^b,
\end{equation}
and
\begin{align*}
\frac{\p}{\p x{}^i}&=\frac{\p \tilde x^j}{\p x{}^i} \frac{\p}{\p \tilde x^j}+ \frac{\p \tilde y^b}{\p x{}^i} \frac{\p}{\p \tilde y^b}= \frac{\p \tilde x^j}{\p x{}^i} \frac{\p}{\p \tilde x^j}+ M^b_{\,\,c,i}(x) \,y^c \frac{\p}{\p \tilde y^b} \\
\frac{\p}{\p y{}^a}&=\frac{\p \tilde x^j}{\p y{}^a} \frac{\p}{\p \tilde x^j}+ \frac{\p \tilde y^b}{\p y{}^a} \frac{\p}{\p \tilde y^b}=M^b_{\,\,a}(x)\, \frac{\p}{\p \tilde y^b}
\end{align*}
where the last equation coincides with  (\ref{alf}). (observe that $\p/\p x^i\colon E\to TE$ thus it is a vector field over $E$, not over $B$).
Replacing in (\ref{vhg})
\begin{align}
\dot{\tilde x}^i&=\frac{\p \tilde x^i}{\p x^j }\,\dot{x}^j, \label{pu1}\\
 \dot{\tilde y}^a&=M^a_{\,b}(x) \, \dot{y}^b+M^a_{\,b,k}(x) \,\dot{x}^k y^b. \label{pu2}
\end{align}
Equations (\ref{bh1}), (\ref{lol}), (\ref{pu1}) and (\ref{pu2}) give the transformation rules for the vector bundle.

\subsection{Vertical space and Liouville vector field}
Let $\pi_B\colon E\to B$ be a vector bundle.
The projection $\pi_E\colon TE\to E$ reads in local coordinates
\[
(x^i,y^a, \dot{x}^i, \dot{y}^a) \mapsto (x^i,y^a) ,
\]
while the map $T\pi_B\colon TE\to TB$ reads
\[
(x^i,y^a, \dot{x}^i, \dot{y}^a) \mapsto (x^i,\dot{x}^i) .
\]
The {\em vertical subbundle} is $VE:= \textrm{ker} \,T\pi_B$, thus it is made of the points
\[
(x^i,y^a, 0, \dot{y}^a) .
\]
Observe that the cocycle for $\dot{y}^a$ becomes
\[
\dot{y}'{}^a=M^a_{\,b}(x) \,\dot{y}^b
\]
which is the same transformation rule for $y^c$. Indeed, we can identify
\[
VE=E\times_B E
\]
through the map $s\colon VE\to E\times_B E$,  $(x^i,y^a, 0, \dot{y}^a) \to (x^i,y^a, \dot{y}^a)$.

The {\em Liouville vector field}
\[
L\colon E\to VE=E\times_B E,
\]
is given by
\[
(x^i,y^a) \mapsto (x^i,y^a, 0, y^a).
\]
Stated in another way, every relation on $E$ which sends $p\in E$ to points on the same fiber, is a subset of $E\times_B E$. The Liouville vector field is the diagonal relation, namely the identity relation, $(x,y) \to (x,y,y)$. Since  $VE$ is included in $TE$, we can regard $L$ as a vector field on $E$. Its expression is
\[
L(x,y)= y^a \frac{\p }{\p y^a}.
\]
This is also the infinitesimal generator of the Lie group action
\[
\mathbb{R}\times E \to E, \quad (t, (x^i,y^a))\mapsto (x^i, e^t y^a) .
\]
Finally, a {\em  soldering form} on a general fibered space is a map (not necessarily a linear isomorphism)
\[
e\colon E\to T^*B\otimes_E VE.
\]
Some fiber bundles admit a natural soldering form. We shall see an example in the next section.

\subsection{The tangent bundle}

If $B=M$, $E=TM$,  and $e_a=e^i_a(x)  \frac{\p}{\p x^i}$ is the chosen base for $T_xM$,  we can introduce the soldering form
\begin{equation} \label{mdp}
e= e^a_i(x) \frac{\p}{\p y^a}\otimes \dd x^i
\end{equation}
where $e^a_i$ is the inverse of $e^i_a$, $\frac{\p}{\p x^i}=e^a_i \frac{\p}{\p y^a}$. The components of this soldering form are called {\em vierbeins}. The transformation rule for the vierbeins is
\[
\tilde e^j_b M^b_{\, a}=\frac{\p \tilde x^j}{\p x^i } \,e^i_a .
\]
It is convenient to make the base $\{e_a\}$ coincide with $\{\p/\p x^i\}$ in any coordinate system, in this way the vierbeins are identity matrices
\[
e^i_b=\delta^i_b, \quad \Rightarrow \ M^j_{\, i}=\frac{\p \tilde x^j}{\p x^i } ,
\]
and under this stipulation we do not need to distinguish the roles of  roman and Latin indices.
The two coordinate vectors in the middle of $(x^i,y^a, \dot{x}^i, \dot{y}^a)$ transform in the same way (they have the same cocycle). Indeed, there is an involution ($j^2=Id$), called {\em canonical flip}, which in coordinates reads
\[
j\colon TTM\to TTM, \quad (x^i,y^i, \dot{x}^i, \dot{y}^i)\mapsto (x^i, \dot{x}^i, y^i, \dot{y}^i).
\]
The symmetrized bundle $STM$ is the subbundle left invariant by $j$, namely that made by those points for which $y^i=\dot{x}^i$.

\begin{remark}
The soldering form  (\ref{mdp}) is closely related with the  {\em canonical endomorphism}
\[
J\colon TTM\to TTM, \quad J=  \frac{\p}{\p y^k} \otimes \dd x^k,
\]
which in coordinates it is given by
\[
 (x^i, y^i, \dot{x}^i,  \dot{y}^i) \to (x^i,  y^i, 0, \dot{x}^i).
\]
It satisfies
\[
\operatorname{Ran}(J)=\operatorname{Ker}(J)=VTM, \qquad \mathcal L_LJ= -J, \qquad J[X,Y]=J[JX,Y]+J[X,JY],
\]
where $\mathcal{L}$ is the Lie derivative. These properties can be used to characterize the tangent bundle among the vector bundles with base $M$ which have twice the dimension of $M$. In the sequel $J$ will not be used although some constructions can pass through it \cite{grifone72}.
\end{remark}

\subsection{Second order equations}
Let us specialize to the case $E=TM$.
A second order equation is a map (for more on $TTM$ see Godbillon \cite{godbillon69})
\[
G\colon TM \to TTM,
\]
with the property that $\pi_{TM}\circ G=Id$, namely it is a section, and $T\pi_M \circ G= Id$, which implies that the map has the form
\[
(x^i,y^i)\mapsto (x^i,y^i, y^i, -2G^i(x,y))
\]
for some function $G^i$ (it is also usual to denote $H^i=-2G^i$). In other words it is a section of the symmetrized bundle. Equivalently,
\[
G(x,y)=y^i\frac{\p }{\p x^i}-2G ^i(x,y)\frac{\p }{\p y^i}.
\]
Thus any integral curve is obtained lifting a solution of
\[
\ddot{x}^i+2 G^i(x,\dot{x})=0.
\]
The cocycle for $-2G$ is the same as that for $\dot{y}$, thus
\begin{align}
 \tilde G^i(\tilde x, \tilde y)&=M^i_{\,j}(x)\, G^j-\frac{1}{2}M^i_{\,j,k}(x) y^k y^j=\frac{\p \tilde x^i}{\p x^j }\,  G^j-\frac{1}{2}\frac{\p^2 \tilde x^i}{\p x^j \p x^k } y^k y^j \nonumber \label{njv}\\
 &= \frac{\p \tilde x^i}{\p x^j } \, G^j(x,y)-\frac{1}{2}\frac{\p^2 \tilde x^i}{\p x^j \p x^k } \frac{\p x^k}{\p \tilde x^l }  \frac{\p x^j}{\p \tilde x^m } \, \tilde y^l \tilde y^m .
\end{align}
Observe that it makes sense to say that this function is quadratic  or positive homogeneous of second degree in $y$, since these conditions are preserved by the cocycle.
Differentiating three times
\[
\frac{\p^3 \tilde G^i}{\p \tilde y^j\p \tilde y^k\p \tilde y^l} =\frac{\p \tilde x^i}{\p x^m } \frac{\p x^p}{\p \tilde x^j} \frac{\p x^q}{\p \tilde x^k} \frac{\p x^r}{\p \tilde x^l} \frac{\p^3 G^m}{\p y^p\p y^q\p y^r}.
\]
The symmetric vector-valued tensor defined for $y\ne 0$, $G\colon (E\backslash 0)\times (TM)^3\to VE$
\[
G^i_{jkl}(x,y) \frac{\p}{\p y^i} \otimes \dd x^j \otimes \dd x^k \otimes \dd x^l, \qquad   G^i_{jkl}(x,y)=\frac{\p^3 G^i}{\p y^j\p y^k\p y^l}
\]
is called {\em Berwald curvature}, and vanishes if and only if the spray is quadratic in $y$.

\begin{remark}
It can be shown assuming positive homogeneity that if $G^i$ were twice continuously differentiable with respect to $y$ at the origin then it would be quadratic (Berwald space). Later on we shall introduce a map $\mathscr{L}\colon TM\to \mathbb{R}$ positive homogeneous of degree two which, again, will be twice  differentiable with respect to $y$ at the origin only in the quadratic case.  Its Hessian $g$ will not be continuously extendable to the zero section in the non-quadratic case. In order to work with more general cases, in several construction we shall remove the zero section, working with the slit tangent bundle $E=TM\backslash 0$ in place of $TM$.
\end{remark}

\subsubsection{Sprays}
A spray $G$ is a second order equation such that
\[
[L,G]=G ,
\]
where, since $VE$ is included in $TE$, we can regard $L$ as a vector field on $E$.
This condition is equivalent to
\[
y^i \frac{\p }{\p y^i}\, G^k(x,y)=2 G^k(x,y) ,
\]
which is the condition of positive homogeneity of degree two: $G(x,sy)=s^2 G(x,y)$ for every $s>0$.
Observe that $G^i$ is positive homogeneous of degree 2 if and only if ${\p^2 G^i}/{\p y^j\p y^k}$ is positive homogeneous of degree zero, which is the case iff
\[
G^{i}_{jkl} y^l=0.
\]

\section{The non-linear connection}

In this section we introduce the non-linear connection of Finsler geometry. Actually, in some special cases the `non-linear connection' can be linear.

\subsection{Connections on vector bundles}
Given a vector bundle $\pi_B\colon E\to B$, we have the exact sequence
\[
0 \to VE \to TE \to E\times_B TB \to 0,
\]
where all the  bundles appearing in the sequence have base $B$. The bundle $E\times_B TB$ is called {\em pullback bundle}\footnote{For the definition of pullback bundle see \cite[Chap.\ 2, Teo.\ 1.7 and p.\ 150,154]{godbillon69}.} and often denoted $\pi^* TB$.
A {\em connection} is a splitting of this exact sequence, namely either a linear inclusion map (linear in the second factor)
\[
\mathcal{N} \colon E\times_B TB \to TE,
\]
called {\em horizontal lift} or a linear {\em vertical projection} (linear on the tangential component, not on the base point on $E$)
\[
\nu\colon TE\to VE,
\]
$\nu\circ \mathcal{N}=0$, the former implying the latter and conversely.

The tangent bundle $E$ gets splitted
\[
TE=HE\oplus VE,
\]
where $HE=\ker \nu=\textrm{Im} \mathcal{N}$.
In coordinates
\begin{equation} \label{ksg}
\mathcal{N}\colon (x^i,y^a, \dot{x}^i)\mapsto (x^i,y^a, \dot{x}^i, -N^a_k(x,y) \dot{x}^k) ,
\end{equation}
that is the vector of $T_xB$
\[
\dot{x}^i\frac{\p}{\p x^i}
\]
is lifted to the vector of $T_{e} E$
\[
\dot{x}^k [\frac{\p}{\p x^k} -N^a_k(x,y) \frac{\p}{\p y^a}],
\]
where the coordinates of $e\in E$ are $(x^i,y^a)$, thus
\begin{equation} \label{vof}
\mathcal{N}= (\frac{\p}{\p x^k} -N^a_k(x,y) \frac{\p}{\p y^a} )\otimes \dd x^k\vert_{TB} .
\end{equation}
From the transformation rules of the first section we find the cocycle for $N^a_j$
\begin{equation} \label{gam}
\tilde N^a_j=(M^a_{\, b} N^b_k- M^a_{\,b,k} y^b)\,\frac{\p x^k}{\p \tilde x^j}.
\end{equation}
It makes sense to consider connections linear in $y$ since the cocycle for $N^a_k$ preserves this condition. It also makes sense to consider connections which are positive homogeneous of degree one in $y$,
\[
N^a_k(x,sy)= s N^a_k(x,y), \quad s>0
\]
since again, this condition is preserved by the cocycle. In this case
\[
N^a_{bk}:=\frac{\p}{\p y^b} N^a_k(x,y)
\]
 is positive homogeneous of degree zero in $y$ thus
\begin{equation} \label{pfc}
y^c\frac{\p}{\p y^c} N^a_{bk}(x,y)=y^c\frac{\p}{\p y^b}  N^a_{ck}(x,y)=0,
\end{equation}
an identity which we shall use later on. It can be useful to observe that a base for $H_eE$, $e\in E$ is given by (with abuse of notation)
\[
\{ \frac{\delta }{\delta x^k}\}, \quad \frac{\delta }{\delta x^k}:=\frac{\p}{\p x^k} -N^a_k(x,y) \frac{\p}{\p y^a}
\]
while a base for $V_eE$ is given by $\{\p/\p y^a\}$. Observe that if we make a change of coordinates $\tilde{x}= \tilde{x}(x)$, then on $TB$, $\frac{\p}{\p \tilde{x}^i}= \frac{\p x^k}{\p \tilde x^i} \frac{\p}{\p {x}^k}$ and since the horizontal lift is  linear
\[
\frac{\delta}{\delta \tilde{x}^i}= \frac{\p x^k}{\p \tilde x^i} \frac{\delta}{\delta {x}^k}.
\]
Let us come to the definition of connection through a projection. The projection in coordinates is
\begin{equation} \label{ome}
\nu\colon (x^i,y^a, \dot{x}^i, \dot{y}^a) \mapsto (x^i,y^a, 0, \dot{y}^a+N^a_k(x,y) \dot{x}^k)
\end{equation}
that is
\begin{equation} \label{con}
\nu=\frac{\p}{\p y^a}\otimes [\dd y^a+N^a_k(x,y) \dd x^k]
\end{equation}
Coming to the dual spaces, $H_eE^*$ is defined as the subspace of $T_e^*E$ whose elements return zero when applied to $V_eE$, and analogously, $V_eE$ is defined  as the subspace of $T_e^*E$ whose elements return zero when applied to $H_eE$. The base of the former space is $\{ \dd x^a\}$ while that for the latter space is
\[
\{\delta y^a\}, \quad \delta y^a:=\dd y^a+N^a_k(x,y) \dd x^k.
\]
Using Req.\ (\ref{lol}) and (\ref{gam}) we find that under a change of coordinates and under a change of base of $E$, these one-forms transform as follows
\[
\delta \tilde y^a=M^a_b(x) \delta y^b.
\]
Clearly, $\delta y^a(\p/\p y^b)=\delta^a_b$, $\dd x^i (\delta/\delta x^j)=\delta^i_j$ and the other contractions vanish. The base for $T_eE$ has commutation relations
\begin{align}
\left[ \frac{\delta}{\delta x^i}, \frac{\delta}{\delta x^j}\right]&= -R^a_{ij} \frac{\p}{\p y^a}, \label{nns}\\
\left[ \frac{\delta}{\delta x^i}, \frac{\p}{\p y^a}\right]&= N^b_{ai}\,\frac{\p}{\p y^b} , \label{nnd}\\
\left[ \frac{\p}{\p y^a}, \frac{\p}{\p y^b}\right]&=0. \label{nnf}
\end{align}
where the vertical map
\[
R\colon E\to VE \otimes_E T^*M\otimes_E T^*M, \quad R=R^a_{ij}(x,y) \frac{\p}{\p y^a} \otimes \dd x^i\otimes \dd x^j,
\]
\begin{equation} \label{cur}
 R^a_{ij}=\frac{\delta N^a_j}{\delta x^i}-\frac{\delta N^a_i}{\delta x^j},
\end{equation}
is called {\em curvature}. We shall see later on that the definition is well posed. For the moment, we just observe that $R=0$ iff the horizontal distribution is integrable, namely iff there exist locally flat sections.

Let $B=M$, $E= TM\backslash 0$, the connection met in this section plays a special role in Finsler geometry and is generically referred to as the {\em non-linear connection}  (which has to be distinguished from the {\em Finsler connections} which we shall introduce later on). The choice of letter $\mathcal{N}$ for the connection recalls  ``Non-linear''.


%
%

\subsection{Covariant derivatives} \label{mfl}
A local section of the bundle $\pi\colon E\to B$, is a map $s\colon U\to E$, defined on an open subset $U\subset B$, such that $\pi\circ s=Id_U$. Any section induces a map $Ts\colon TU\to TE$
\[
Ts\colon (x^i,\xi^i) \mapsto (x^i, s^a(x), \xi^i,  \frac{\p s^a}{\p x^k}  \, \xi^k)
\]
which can be composed with the projection defining the connection $\nu\colon TE\to VE$ so as to obtain the covariant derivative (we shall use $\nabla$ for connections which are known to be linear)
\[
D s= \nu\circ Ts\colon TU\to VE.
\]
Using Req.\ (\ref{ome}) we find its expression in coordinates
\[
D s\colon (x^i,\xi^i) \mapsto  (x^i, s^a(x), 0, D_\xi s^a)
\]
where
\[
 D_{\xi} s^a=\frac{\p s^a}{\p x^k} \, \xi^k+N^a_k(x, s(x)) \xi^k .
\]
Although the covariant derivative has been defined on local sections it is clear that the section needs to be defined just over some curve on $B$ provided $\xi$ is chosen tangent to the curve.

Finally, it is important to observe that  the non-linear connection is positive homogeneous of degree one, iff for every positive function $f\colon U\to \mathbb{R}$
\[
D_\xi (f s)^a=(\p_\xi f) s^a+f D_\xi s^a .
\]

\begin{remark} \label{mdl}
If $E$ is a vector bundle then $VE=E\times_B E$ and we can further project on the second factor
\[
TU \xrightarrow{Ts} TE \xrightarrow{\nu} VE=E\times_B E\xrightarrow{\pi_2} E
\]
thus for any given section the covariant derivative sends an element of the vector bundle to another element with the same base point
\[
(x^i,\xi^i) \mapsto  (x^i,  D_\xi s^a).
\]
\end{remark}

If $E=TM\backslash 0$, $B=M$, then $s\colon U \to E$ is a nowhere vanishing vector field.
We say that a nowhere vanishing vector field $V\colon x(I)\to TM\backslash 0$, defined over the image of a curve $x\colon I\to M$, is {\em parallelly transported} if $D_{\dot{x}} V=0$. Since $E$ is linear it is also possible to consider {\em linear connections} which are those which in coordinates read
\[
N^a_k(x, y)=N^a_{jk}(x)\,y^j,
\]
in which case we recover the usual (Koszul) covariant derivative.\footnote{The placement of the indices has been chosen so as to recover the formula $\tau_{ij}=D_{e_i} e_j-D_{e_j} e_i-[e_i,e_j]$ for the torsion and $D_{e_i}e_j=N^k_{ji} e_k$, in the linear connection case. Our conventions are  the same of \cite{misner73} also for what concerns wedge product and exterior differential (Spivak \cite{spivak79}), which are different from those of Kobayashi-Nomizu \cite{kobayashi63}.}

Actually, due to the existence of the involution $j\colon TTM\to TTM$ we can also define a map
\[
\overset{flip}{D}
 s=\nu\circ j \circ Ts\colon TU\to VE.
\]
which despite the notation is not really a covariant derivative as it is non-linear in the entry which selects the direction of differentiation.
Using Eq.\ (\ref{ome}) we find its expression in coordinates
\[
\overset{flip}{D}s\colon (x^i,\xi^i) \mapsto  (x^i, \xi^i, 0, \overset{flip}{D}_\xi s^a)
\]
where
\[
 \overset{\!\!\! flip}{D_{\xi}} s^a=\frac{\p s^a}{\p x^k} \, \xi^k+N^a_k(x, \xi) s^k .
\]
For every function $f\colon U\to \mathbb{R}$ the flipped covariant derivative satisfies
\[
\overset{\!\!\! flip}{D_{\xi}} (f s)^a=(\p_\xi f) s^a+f \overset{\!\!\! flip}{D_{\xi}} s^a .
\]
The flipped-parallel transport is obtained imposing $\overset{\!\!\! flip}{D_{\dot{x}}} V=0$ along the curve.
Here the parallel transport  is independent of the parametrization of the curve iff   $\mathcal{N}$ is positively homogeneous of degree one,   which is  usually the case in  Finsler space theory. Indeed, some author working in Finsler geometry refer to this notion as the {\em  covariant derivative} \cite[Eq.\ (5.33)]{shen01}. In what follows we shall not use this covariant derivative.

\subsection{Curvature and torsion of the non-linear connection}
In a manifold $E$ with local coordinates $\{ z^\mu\}$ a vector valued form reads\footnote{As mentioned, we use the Spivak \cite{spivak79} convention $\alpha\wedge \beta =\alpha \otimes \beta -\beta \otimes \alpha$, for any two one-forms $\alpha$ and $\beta$ (compare with Kobayashi-Nomizu). This convention allows one to write $F=\dd A$ as the relation between the electromagnetic field and potential.}
\begin{align*}
K&\colon E\to \wedge^k\, T^*E \otimes TE\\
K&=K^{\mu}_{\lambda_{1} \ldots \lambda_{k}} \frac{1}{k!} \, d
z^{\lambda_{1}} \wedge \ldots \wedge d z^{\lambda_{k}} \otimes
\p_{\mu}.
\end{align*}
The Fr\"olicher-Nijenhuis bracket of vector valued forms can be defined through the  coordinate expression
\begin{align*}
[K,L]&=(K^{\nu}_{\lambda_{1} \ldots
\lambda_{k}} \p_{\nu} L^{\mu}_{\lambda_{k+1} \ldots
\lambda_{k+l}} \!\!-(-1)^{kl}L^{\nu}_{\lambda_{1} \ldots
\lambda_{l}} \p_{\nu} K^{\mu}_{\lambda_{l+1} \ldots
\lambda_{k+l}} \!\!\\ &-k K^{\mu}_{\lambda_{1} \ldots \lambda_{k-1}
\nu}
\p_{\lambda_{k}} L^{\nu}_{\lambda_{k+1} \ldots \lambda_{k+l}}\\
{}&+(-1)^{kl} l L^{\mu}_{\lambda_{1} \ldots \lambda_{l-1} \nu}
\p_{\lambda_{l}}K^{\nu}_{\lambda_{l+1} \ldots
\lambda_{k+l}})\frac{1}{(k+l)!} \, d z^{\lambda_{1}} \wedge \ldots
\wedge d z^{\lambda_{k+l}} \otimes \p_{\mu}
\end{align*}
It generalizes the usual Lie brackets for vector fields.

The curvature of a bundle $\pi_B \colon E \to B$ endowed with a connection $\nu\colon TE\to VE$ is defined through\footnote{With respect to the reference we modify the coefficient so as to obtain the usual expression in the Riemannian case. Our choice coincides with the more recent choice adopted by Modugno and collaborators.} \cite{modugno91}
\[
R=-[\mathcal{N},\mathcal{N}].
\]
This is a slight abuse of notation since $\mathcal{N}\colon E\to T^*B\otimes TE$, so $\mathcal{N}$ appearing in the previous formula is really $\pi_B^*\mathcal{N}$.
Since $\mathcal{N}$ is a  vector valued one-form
\begin{align*}
R^\mu_{\alpha \beta}&=-\frac{1}{2}[\mathcal{N}^\nu_{\alpha} \p_\nu \mathcal{N}^\mu_{\beta}+\mathcal{N}^\nu_\alpha \p_\nu \mathcal{N}^\mu_\beta-\mathcal{N}^\mu_{\nu}\p_\alpha \mathcal{N}^\nu_\beta-\mathcal{N}^\mu_\nu \p_\alpha \mathcal{N}^\nu_\beta -(\beta/\alpha)]\\
& =- [\mathcal{N}^\nu_{\alpha} \p_\nu \mathcal{N}^\mu_{\beta} -\mathcal{N}^\mu_\nu \p_\alpha \mathcal{N}^\nu_\beta-\mathcal{N}^\nu_{\beta} \p_\nu \mathcal{N}^\mu_{\alpha} +\mathcal{N}^\mu_\nu \p_\beta \mathcal{N}^\nu_\alpha] ,
\end{align*}
where $(\beta/\alpha)$ stands for ``plus terms with $\beta$ an $\alpha$ exchanged''.

In our case the coordinates are $(x^1,\cdots, x^n, y^1, \cdots, y^n)$, so let us introduce a notation according to which $z^i=x^i$ and $z^a=y^a$.
From Eq.\ (\ref{vof}) we read the components $\mathcal{N}^i_j=\delta^i_j$, $\mathcal{N}^a_b=\mathcal{N}^i_a=0$ while $\mathcal{N}^a_k=-N^a_k$. This makes several terms vanish and we are left with
\[
R=-(\p_i \mathcal{N}^a_{j} - \p_j \mathcal{N}^a_{i} + \mathcal{N}^b_{i} \p_b \mathcal{N}^a_{j}   -\mathcal{N}^b_{j} \p_b \mathcal{N}^a_{i} ) \, \frac{\p}{\p y^a}\otimes \dd x^i \otimes \dd x^j
\]
Thus
\[
R=( \frac{\p}{\p x^i} N^a_j - \frac{\p}{\p x^j} N^a_i  +N^b_{j} N^a_{bi}  -N^b_{i} N^a_{bj} ) \,\frac{\p}{\p y^a}\otimes \dd x^i \otimes \dd x^j
\]
which coincides with Eq.\ (\ref{cur}).

While every connection admits a curvature, in order to define the torsion it is necessary to specify a soldering form \cite{modugno91}.
The torsion of a connection $\mathcal{N}$ with respect to a soldering form $e\colon E\to T^*B\otimes_E VE$ is
\begin{equation} \label{tor}
\tau=2 [\mathcal{N},e] \ \colon E \to \Lambda^2 T^*B\otimes_E VE
\end{equation}
In components it is given by
\[
\tau=[\p_i e_j^a+\mathcal{N}^b_j \frac{\p}{\p y^b} e^a_{j}-(\frac{\p}{\p y^b} \mathcal{N}^a_i ) e^b_j-(i/j)] \,\frac{\p}{\p y^a}\otimes \dd x^i \otimes \dd x^j.
\]
Since $E=TM\backslash 0$ we can use as soldering form the canonical one with coordinates $e^a_i=\delta^a_i$, so that
\[
\tau=(N_{ji}^k-N_{ij}^k) \,\frac{\p}{\p y^k}\otimes \dd x^i \otimes \dd x^j.
\]
Observe that if $\mathcal{N}$ is linear we have for every $X,Y\colon M\to TM$
\[
\tau(X,Y)=D_X Y-D_Y X-[X,Y]
\]
since this formula holds whenever $X$ and $Y$ are replaced by elements of the holonomic base $\p/\p x^i$.

\subsection{Relationship between non-linear connections and sprays}

In this section let $E=TM\backslash 0$, $B=M$.

\begin{definition}
Let $\eta:I \to M$, $I\subset \mathbb{R}$, be a curve with non-vanishing tangent vector (regular) $\dot{\eta}\colon I \to TM$, then $s=\dot{\eta}\circ \eta^{-1}$ is a section of $E$ defined over a curve and we say the curve is a {\em geodesic} if  $D_s s=0$.
\end{definition}

\begin{theorem}
Let $G$ be a spray and let $\mathcal{N}$ be a connection, then the map
\[
T\colon E\times_B TB \to VE
\]
defined locally by
\begin{equation} \label{pol}
T=T^i_{\, j}(x,y) \frac{\p}{\p y^i}\otimes \dd x^j, \qquad T^i_{\, j}= N^i_j-\frac{\p G^i}{\p y^j},
\end{equation}
is globally well defined.
\end{theorem}

\begin{proof}
The cocycle for $G$ has been calculated in Eq.\ (\ref{njv})
\begin{align*}
 \tilde G^i(\tilde x,\tilde y) &= \frac{\p \tilde x^i}{\p x^j } \, G^j(x,y)-\frac{1}{2}\frac{\p^2 \tilde x^i}{\p x^j \p x^k } \frac{\p x^k}{\p \tilde x^l }  \frac{\p x^j}{\p \tilde x^m } \,\tilde y^l \tilde y^m .
\end{align*}
Differentiating
\[
\frac{\p \tilde G^i}{\p \tilde y^s} =\frac{\p \tilde x^i}{\p x^j } \frac{\p x^k}{\p \tilde x^s} \frac{\p G^j}{\p y^k}-\frac{\p^2 \tilde x^i}{\p x^j \p x^k }  \frac{\p x^j}{\p \tilde x^s}\, y^k.
\]
which is the same cocycle of $N^i_s$, thus the difference has the right cocycle which makes the definition independent of the  coordinate system.
\end{proof}

 Let us compose the Liouville field $L\colon E \to E \times_{M} TM$ with the connection $\mathcal{N}$. We obtain in coordinates
\[
\mathcal{N}\circ L\colon (x^i,y^i) \mapsto (x^i,y^i,y^i) \mapsto  (x^i,y^i, y^i, -N^i_k(x,y) y^k).
\]
We recognize the map defining a second order differential equation with
\begin{equation} \label{vkf}
2G^i(x,y)=N^i_k(x,y) y^k .
\end{equation}

\begin{theorem}
Let $G$ be a spray and let $\mathcal{N}$ be a connection. The geodesics for $\mathcal{N}$ coincide with the paths for the spray if and only if $G=\mathcal{N}\circ L$.
\end{theorem}

Given a connection $\mathcal{N}$, the spray $G=\mathcal{N}\circ L$, will be called {\em the spray of the connection}.

\begin{proof}
Let $x(t)$ be a regular $C^1$ curve which solves the spray equation $\ddot x^i+2G^i(x,\dot{x})=0$ with initial conditions $x(t)=x_0$, $\dot{x}(0)=\dot{x}_0$. We have
\[
D_{\dot{x}}\dot{x}^i=\ddot{x}^i+N^i_k(x,\dot{x}) \dot{x}^k,
\]
thus it is a geodesic only if  at the origin $2G^i(x_0,\dot{x}_0)=N^i_k(x_0,\dot{x}_0) \dot{x}_0$ and since the initial conditions are arbitrary this equation must hold generally, that is $\mathcal{N}\circ L=G$. Conversely, if this equation holds it is immediate that the geodesics and spray paths coincide.
\end{proof}

\begin{theorem} \label{vkk}
Let $G$ be a spray, then   there is a connection  whose spray is $G$ ($G=\mathcal{N}\circ L$), e.g.\ that defined in coordinates by
\begin{equation} \label{blq}
N^i_k(x,y)= \frac{\p G^i}{\p y^k}.
\end{equation}
Let $G$ be a spray, and let $\mathcal{N}$ be a connection whose spray is $G$, then there is a map $T\colon E\times_B TB \to VE$ such that
 $T\circ L=0$, namely such that
\begin{equation} \label{ndo}
T^i_{\, j}(x,y) y^j=0,
\end{equation}
and
\begin{equation} \label{blp}
N^i_k(x,y)= \frac{\p G^i}{\p y^k}+T^i_{\, k}.
\end{equation}
 $T$ is positive homogeneous of degree one in $y$ iff  $\mathcal{N}$ is.
\end{theorem}

Observe that the map (imposing positive homogeneity) ``connection $\to$ spray'' is surjective but not injective.

\begin{proof}
That the cocycles of $N^i_k$ and $G^i$ are the same  has been proved in the previous theorem. For the second statement as $\mathcal{N}$ has spray $G$ we have $N^i_{k} y^k=2 G^i$, and by the positive homogeneity of the spray $N^i_{k} y^k=\frac{\p G^i}{\p y^k} \,y^k$ which implies that $T^i_k$ satisfies Eq.\ (\ref{ndo}).
\end{proof}

We have already shown that the map  $\tau\colon E\times_B TB\otimes_B TB \to VE$ given locally by
\[
\tau=(N^k_{ji}- N^k_{ij}) \,\frac{\p}{\p y^k}\otimes \dd x^i \otimes \dd x^j
\]
 is the {\em torsion} of the non-linear connection.
If the spray of the connection is $G$, then according to Theorem \ref{vkk}
\[
\tau=[\frac{\p}{\p y^j} T^k_i-\frac{\p}{\p y^i} T^k_j] \,\frac{\p}{\p y^k}\otimes \dd x^i \otimes \dd x^j.
\]
where $T^k_j$ is defined by Eq.\ (\ref{blp}).

\begin{proposition}
Every (non-linear) connection positive homogeneous of degree one is determined  by its spray  and by its torsion as follows
\begin{equation} \label{bln}
N^i_k(x,y)= \frac{\p G^i}{\p y^k}-\frac{1}{2}\,\tau^i_{jk} y^j,
\end{equation}
and $T^k_j=-\frac{1}{2}\tau^k_{ij} y^i$.
\end{proposition}

\begin{proof}
 As the connection is positive homogeneous of degree one, then so is $T^k_i$ and we have
\[
(\frac{\p}{\p y^i} T^k_j+\frac{\p}{\p y^j} T^k_i) y^i=T^k_j+\frac{\p(T^k_i y^i)}{\p y^j} -T^k_i \delta^i_j=0
\]
Thus
\[
T^k_j=\frac{1}{2}[-\tau^k_{ij}+\frac{\p}{\p y^i} T^k_j+\frac{\p}{\p y^j} T^k_i] y^i=-\frac{1}{2}\tau^k_{ij} y^i.
\]
Finally, Eq. (\ref{blp}) reads as Eq.\ (\ref{bln}), which proves the  proposition.
\end{proof}

\section{The linear connection and the metric}

In this section we introduce the linear connections of Finsler geometry.

\subsection{Finsler connections and covariant derivatives of fibered morphisms}

Let $B=M$, $E=TM\backslash 0$. We are interested in objects of the form
\[
O\colon E \to E\times_M \overbrace{TM \otimes_M \cdots \otimes_M TM}^{a}  \otimes_M \overbrace{T^*M\otimes_M \cdots \otimes_M T^*M}^{b}
\]
where we demand that the first map on $E$ be just the identity (this condition is notationally expressed removing the first factor $E\times$ and replacing the instances of $\otimes_M$ with $\otimes_E$). That is these maps are fibered morphisms: $\pi_B\circ O=\pi_B$.

We call these maps  {\em tensors} of type $(a,b)$. They arise naturally in the theory of connections because the vertical space $VE$ is naturally isomorphic to $E\times_M TM$, thus any map of the form
\[
O'\colon E\to   VE \otimes_E \cdots \otimes_E VE \otimes_E T^*M\otimes_E \cdots \otimes_E T^*M
\]
can actually be regarded as a tensor. The transformation rule of a tensor is the usual one, namely the Jacobian $\p \tilde{x}/\p x$ appears several times together with its inverse transpose depending on the type of tensor.
The only difference is that the components of our tensors depend on an argument belonging to $E$ rather than $M$. Fibered morphisms are generalizations of the concept of section $s\colon M\to (TM)^a\otimes (T^*M)^b$, indeed from a section we obtain a corresponding fibered morphism through $O=s\circ \pi_B$.
Some authors prefer to characterize these objects through their transformation properties.  They are called d-objects, $M$-objects or Finsler objects.

%
%

The Fr\"olicher-Nijenhuis bracket  provides a derivation in the graded algebra of vertical valued forms. These forms belong to the larger family of tensors of type $(1,b)$.

Let us introduce the { Finsler connections}  \cite{akbarzadeh63,akbarzadeh88,abate94,anastasiei96}.
Let $B=M$ and $E=TM\backslash 0$.

\begin{definition}
A {\em Finsler connection}  is a pair $(\mathcal{N},\nabla)$ where $\mathcal{N}$ is a non-linear connection on the bundle $\pi_B\colon E\to B$, namely a projection $\nu\colon TE\to VE$, and $\nabla$ is a linear connection on the bundle $\pi_E\colon \tilde{E}\to E$, namely a projection $\tilde\nu\colon T\tilde{E}\to V\tilde{E}$, where $\tilde{E}=VE$. Furthermore, $\mathcal{N}$ and $\nabla$ must be positive homogeneous of degree one with respect to the fiber of $E$ (see Eq.\ (\ref{cla}) for a clarification).
\end{definition}

\begin{remark}
We stress that $\nabla$ does not act over arbitrary sections $\check{X}:E\to TE$ but only over those sections with image in $VE$, and also returns sections of the  same vector bundle. Some authors \cite{miron94} work instead with linear (Koszul) connection on the bundle $TE \to E$. However, this approach requires to impose additional conditions on the linear connection and introduces  local coefficients which in the end are not really necessary.
\end{remark}

Let us consider a map $f:E\to E$ with the property, $\pi_B\circ f=\pi_B$ or, which is the same,  a section (denoted in the same way) $f\colon E\to \tilde{E}$
\[
f\colon (x^i,y^a)\mapsto (x^i,y^a,f^a(x,y)).
\]
The Liouville vector field provides an example, $f^a=y^a$. These sections of $VE$ are important since they provide a special type of non-multi-valued relation on $E \times_M E$, for which the Liouville field is the identity.
Accordingly to Remark \ref{mdl} with $\tilde{E}$ in place of $E$, and $E$ in place of $B$, each section  $f\colon E \to \tilde{E}$
implies a map $\nabla f\colon TE \to \tilde{E}$, given in coordinates by
\[
\nabla f\colon (x^i,y^a,\dot{x}^i,\dot{y}^a) \mapsto (x^i,y^a,\nabla_{\dot{x}+\dot{y}} f^a) ,
\]
where $\dot{x}+\dot{y}$ is a short-hand for $\dot{x}^i\frac{\p}{\p x^i}+\dot{y}^a \frac{\p }{\p y^a}$. It is possible to regard this vector as the sum of an horizontal and vertical part, thus we can write
\[
\nabla f^a=(\nabla_{\delta/\delta x^i} f^a)\, \dd x^i+(\nabla_{\p/\p y^b} f^a) \,\delta y^b.
\]
Actually, here it is convenient to consider the composition
\[
{\nabla} f\circ \mathcal{N}\colon E\times TM \to VE, \qquad  (x^i,y^a,\dot{x}^i) \mapsto (x^i,y^a,\nabla_{\mathcal{N}(\dot{x})} f^a) ,
\]
which acts only on the horizontally lifted vectors (here \[
\mathcal{N}(\dot{x})=\dot{x}^i[\frac{\p}{\p x^i}-N^a_i(x,y)\frac{\p}{\p y^a}], \quad )
\] and
\[
{\nabla} f\circ \nu \colon TE \to VE, \qquad  (x^i,y^a,\dot{x}^i,\dot{y}^a) \mapsto (x^i,y^a,\nabla_{[\dot{y}^a+N^a_k \dot{x}^k]\frac{\p}{\p y^a}} f^a)
\]
which is essentially the action of the covariant derivative on the vertical vectors.
We set (here we use the linearity of $\nabla$)
\begin{align*}
\nabla^H f\colon & E\times_M TM \to E,  \\
\nabla_{\dot{x}}^H f^a &:=\nabla_{N(\dot{x})} f^a(x,y)=\frac{\p f^a}{\p x^i}\dot{x}^i-N^b_i(x,y)\frac{\p f^a}{\p y^b}\dot{x}^i+H_{bi}^a (x,y)f^b \dot{x}^i ,
\end{align*}
where
\[
H_{bi}^a \frac{\p}{\p y^a}=\nabla_{\delta/\delta x^i} \frac{\p}{\p y^b}.
\]
Since $\p/\p y^b$, $\delta/\delta x^i$ transform as tensors, the coefficients $H^a_{bi}$ have the typical cocycle of connections coefficients on $M$
\[
\tilde H^a_{b i} =\frac{\p \tilde x^a}{\p x^c } \frac{\p x^d}{\p \tilde x^b} \frac{\p x^j}{\p \tilde x^i} H^c_{d j} +\frac{\p^2  x^c}{\p \tilde x^b  \p  \tilde x^i }  \frac{\p \tilde x^a}{\p  x^c}  .
\]
Similarly, we set
\begin{align*}
\nabla^V \! f & \colon E\times_M TM=VE \to E,  \\
\nabla^V_{\dot{x}} f^a(x,y)&:=\nabla_{\dot{x}^c e_c} f^a=\dot{x}^c\frac{\p f^a}{\p y^c}+V^a_{bc}(x,y) f^b \dot{x}^c ,
\end{align*}
where
\[
V^a_{bc}\frac{\p}{\p y^a}=\nabla_{{\p}/{\p y^c}} \frac{\p}{\p y^b}.
\]
For shortness we shall also write  $\nabla^H_i$ for $\nabla^H_{\p/\p x^i}$, and  $\nabla^V_i$ for $\nabla^V_{\p/\p x^i}$. It is also common to use a short bar for $\nabla^H$ and a long bar for $\nabla^V$, e.g.\
\[
f^a_{ \vert i} :=\nabla^H_i f^a, \qquad f^a \vert_i :=\nabla^V_i f^a.
\]
From $\nabla$ we have obtained derivatives $\nabla^H$, $\nabla^V$, which act on maps $f\colon  E \to E$ which are fibered morphisms,  $\pi_B\circ f=\pi_B$, returning maps of the same type. Similar definitions with $r:E\to \mathbb{R}$ a real function show that
\[
\nabla^H_i r=\frac{\delta}{\delta x^i} r, \quad \nabla^V_i r= \frac{\p}{\p y^i} r.
\]
Since these derivatives of maps to $E$ are linear they can be extended to tensors in the usual way, namely imposing that they respect the coupling between one-forms and vectors and then extending to tensor products. For instance, over a tensor $g_{a b}(x,y)$ of type $(0,2)$ they act as follows
\begin{align}
\nabla^H_{i} \,g_{jk} &=\frac{\delta}{\delta x^i}\, g_{jk}-H^l_{ji} \, g_{lk}-H^l_{ki}\, g_{jl} , \label{vlf}\\
\nabla^V_{d}\, g_{ab} &=\frac{\p}{\p y^d}\, g_{ab}-V^c_{ad} \, g_{cb}-V^c_{bd}\, g_{ac}. \label{vlg}
\end{align}

It can be useful to observe that with respect to the holonomic base $(\p/\p x^i, \p/\p y^a)$ the connection coefficients of $\nabla$ are given by $V^a_{bc}$ and $\check V^a_{bi}$ where
\[
\check V^a_{bi}=H^a_{bi}+N^c_i V^a_{bc}, \qquad \check V^a_{bi} \frac{\p}{\p y^a}=\nabla_{\p/\p x^i}\frac{\p}{\p y^b}.
\]
The cocycle for $V^a_{bc}(x,y)$ requires some comments. If we make a transformation of coordinates on $E$ then the coefficients of $\nabla$, namely $(\check V^a_{bi},V^a_{bc})$  transform with the usual non-linear cocycle for the connection coefficients, however, under the more restricted group of coordinate changes induced by a change of coordinates on $M$, the coefficients   $V^a_{bc}$ transform as a $(1,2)$ tensor, because $\frac{\p}{\p y} \frac{\p \tilde x}{\p x}=0$. Also recall that since $E=TM$, the coordinate transformations in which we are interested preserve the linear structure of the fiber, see Eq. (\ref{lol}). The Finsler connection is locally determined by the triple of coefficients
\[
(N^a_i, H^a_{bi}, V^a_{bc}).
\]
The positive  homogeneity condition on the  Finsler connection reads
\begin{equation} \label{cla}
N^i_b(x,sy)=sN^i_b(x,y), \quad H^a_{bi}(x,sy)= H^a_{bi}(x,y),\quad V^a_{bc}(x,sy)=V^a_{bc}(x,y),
\end{equation}
for every $s>0$.

We are now going to place a compatibility condition between the connections $\mathcal{N}$ and $\nabla$. Let $TE=HE\oplus VE$ accordingly to the splitting due to $\mathcal{N}$. We demand the following condition called {\em regularity}
\begin{equation} \label{ldg}
HE=\textrm{Ker} \,{\nabla} L
\end{equation}
where $L$ is the Liouville vector field ($f^a(x,y)=y^a$). Equivalently, for any $e\in E$ the map
\[
{\nabla} L(e)\colon T_eE\to V_eE, \qquad v\mapsto {\nabla}_v L(e)
\]
must have maximum rank, which implies that its kernel splits $TE$ as sum of  two subspaces of the same dimension, where the kernel is identified with $HE$.
The condition (\ref{ldg}) reads
\begin{equation} \label{mkk}
0=\nabla^H_{i} y^a=-N^a_i+H_{bi}^a (x,y)y^b.
\end{equation}
Since ${\nabla} L(e)$ has maximal rank it must be an invertible linear isomorphism if restricted to $VE$, thus
\begin{equation}
\det \nabla^V_c y^a=\det(\delta^a_c+V^a_{bc}(x,y) y^b)\ne 0.
\end{equation}
If we are given just  ${\nabla}$ then we say that this linear connection is regular if the last equation holds (Abate and Patrizio \cite{abate94} call them {\em good} connections). In this case \cite{akbarzadeh63,akbarzadeh88} the splitting of $TE$ induced by the kernel of this map provides the missing non-linear connection $\mathcal{N}$.
Thus, under the assumption of regularity  the geometry of Finsler spaces is determined by a covariant derivative $\nabla$ on $E$.

Observe, that every non-linear connection positive homogeneous of degree one can be obtained in this way, just set
\begin{align}
H^a_{bi}&= N^a_{bi}, \label{pfv} \\
V^a_{bc}&=0,
\end{align}
and define $\nabla$ through these coefficients. Given a general regular Finsler connection $(\mathcal{N},\nabla)$ we cannot conclude that Eq.\ (\ref{pfv}) holds, however, by positive homogeneity, multiplying Eq.\ (\ref{mkk}) by $y^c \p/\p y^c$
\[
y^c(H^a_{ci}- N^a_{ci})=0.
\]

\begin{proposition}
Let $t\mapsto (x(t),y(t))$ be an integral curve of the spray of $\mathcal{N}$, then it satisfies $\nabla_{\dot{x}}^H y=0$ iff $2 G^a=N^a_i(x,y) y^i=H^a_{bi} y^a y^i$, in particular this is true for the regular Finsler connections.
\end{proposition}

\begin{proof}
For the integral curve of a spray $y=\dot{x}$.
\end{proof}

To summarize we can place the following three conditions on a Finsler connections
\[
V^a_{bc}(x,y) y^b=0, \ \Rightarrow \ N^a_i(x,y)=H^a_{bi}(x,y) y^b, \ \Rightarrow  \ N^a_i(x,y) y^i=H^a_{bi} y^a y^i.
\]
The middle one is called regularity but, as we shall see, most interesting Finsler connection satisfy already the first one which is equivalent to
\begin{equation} \label{hzp}
\nabla_v L= \nu (v).
\end{equation}

We can already introduce a very interesting Finsler connection.

\begin{definition} \label{clp}
The {\em canonical Finsler connection} of a non-linear connection, or the {\em Berwald connection}, is determined by the local triple
\[
(N^a_i, N^a_{bi},0).
\]
Since every spray induces a canonical non-linear connection (Theorem \ref{vkk}) we can also define a  {\em canonical Finsler connection} of a spray  as follows
\[
(G^i_k, G^i_{jk}, 0), \quad \textrm{where } G^i_k=\frac{\p G^i}{\p y^k} \textrm{ and } G^i_{jk}=\frac{\p^2 G^i}{\p y^j\p y^k}.
\]
\end{definition}
This definition is well posed as it is independent of the coordinate system. By positive homogeneity this Finsler connection is regular and satisfies (\ref{pfv}).
 It is remarkable  that with just the notion of `geodesic' we can define a covariant differentiation over the `configuration space' $TM$.

\subsubsection{Pullpack of the linear connection}

Let us mention an approach which reduces a Finsler connection to a more familiar (Koszul) linear connection on the linear bundle $\pi_M \colon TM\to M$.

In the previous section we have introduced a linear connection on the bundle $\pi_E\colon VE \to E$, $E=TM\backslash 0$. Given a section $s\colon M\to E$ the pullback bundle $s^* (VE)$ is the usual tangent bundle $\pi_M\colon TM\to M$. On it we can introduce a pullback linear connection $\overset{s}{\nabla}:=s^*\nabla$, defined as usual as follows. Given  sections $X,Y\colon M\to TM$, we calculate $s_*(X)\in TE$, introduce the section of $\pi_E\colon VE \to E$ defined just over $s(M)$ given by $\tilde Y(s(p))=(s(p), Y(p))$, $p\in M$, where we used the identification $VE=E\times_M E$, we arbitrarily extend $\tilde Y(x,y)$ in a neighborhood of $s(M)$ so that $\tilde Y(x,s(x))=(s(x),Y(x))$, and finally  we set
\[
\overset{s}{\nabla}_X Y:=s^*(\nabla_{s_*(X)}  \tilde Y)=s^*(\nabla^V_{D_Xs} \tilde Y+\nabla^H_X \tilde Y).
\]
In components this covariant derivative reads
\[
\dd x^k(\overset{s}{\nabla}_{\p/\p x^i} \frac{\p}{\p x^j})(x)= V^k_{ja}(x,s(x))[\frac{\p s^a}{\p x^i}+N^a_i(x,s(x))]+H^{k}_{ji}(x,s(x)).
\]
If $Y=s$ then the extension $\tilde{Y}$ can be chosen to be $L$ thus $\overset{s}{\nabla}_X s=s^*(D_X s)=s^*(\nabla^H_X s)$ by Eq.\ (\ref{hzp}). If $X=s$ is geodesic then $\overset{s}{\nabla}_s Y=s^*(\nabla^H_s \tilde Y)$.

A vector field $Y$ defined along a curve $x\colon I \to M$ is parallely transported iff $\overset{s}{\nabla}_{\dot{x}} Y=0$ which in coordinates reads
\[
\frac{\dd Y}{\dd t}+ \left\{V^k_{ja}(x,s(x))[\frac{\p s^a}{\p x^i}+N^a_i(x,s(x))]+H^{k}_{ji}(x,s(x)) \right\}\dot{x}^i Y^j=0.
\]
If $\nabla$ is the canonical Finsler connection of a spray then
\[
\overset{\dot{x}}{\nabla}_{\dot{x}}Y=\overset{flip}{D}_{\!\!\dot{x}} Y .
\]
For what concerns the curvature $\overset{s}{R}$ of $\overset{s}{\nabla}$, it is just the pullback of the curvature of the linear connection $\nabla$ which will be calculated in Sect.\ \ref{cuh}.

\subsection{The Finsler Lagrangian and its spray}

Let $\mathscr{L}\colon TM\backslash 0 \to \mathbb{R}$ be a function such that
\begin{equation} \label{mod}
\mathscr{L}(x,sy)=s^2 \mathscr{L}(x,y),
\end{equation}
and such that
\[
g_{ij}(x,y)=\frac{\p^2 \mathscr{L}}{\p y^i \p y^j}
\]
is non-singular.\footnote{Many authors use to work with $F=\sqrt{2\mathscr{L}}$, thus excluding without special reasons the  non-positive definite case.} The pair $(M,\mathscr{L})$ is called {\em pseudo-Finsler space}, and $\mathscr{L}$ is called {\em Finsler Lagrangian}.
Through the matrix $g$ we define the (Finsler) metric
\[
g\colon TM\backslash 0\to T^*M\otimes T^*M, \qquad g=g_{ij}(x,y) \, \dd x^i\otimes \dd x^j.
\]
The space is {\em Finsler} if $g$ is positive definite. The function $\mathscr{L}$ will be considered sufficiently differentiable for our purposes, typically $C^5$ would be fine. It can be shown using positive homogeneity that $\mathcal{L}$ is extendible as a $C^1$ function on the zero section. It can be extended as a $C^2$ function if and only if it is quadratic in $y$ (pseudo-Riemannian case).

The Euler-Lagrange equations for the action $\int \mathscr{L}(x,\dot{x})\,\dd t$ with fixed extremes are
\[
0=\frac{\dd}{\dd t} \frac{\p \mathscr{L}}{\p y^j}- \frac{\p\mathscr{L} }{\p x^j }=\frac{\p^2\mathscr{L} }{\p y^i \p y^j} \, \dot{y}^i+\frac{\p^2\mathscr{L} }{\p x^i \p y^j} \,y^i- \frac{\p\mathscr{L} }{\p x^j }, \qquad y^i=\dot x^i,
\]
which can be rewritten in the form
\[
\ddot x^i+2\mathcal{G}^i(x,\dot{x})=0,
\]
where
\begin{align}
2 \mathcal{G}^i(x,y)&= g^{is}\left( \frac{\p^2\mathscr{L} }{\p x^k \p y^s} \,y^k -\frac{\p\mathscr{L} }{\p x^s } \right) \label{axo}\\
&=\frac{1}{2}\, g^{is} \left( \frac{\p}{\p x^j} \,g_{sk}+\frac{\p}{\p x^k} \, g_{sj}-\frac{\p}{\p x^s}\, g_{jk}\right) y^j y^k , \label{axu}
\end{align}
and $g^{is}$ is the inverse metric $g^{is} g_{sk}=\delta^i_k$.  The coefficients $\mathcal{G}^i$ of Eq.\ (\ref{axu}) define a spray, this can be verified checking the transformation rule, though it should be clear from the variational origin of these second order equations that their definition is independent of the chart.

 The metric $g$ is positive homogeneous of degree zero in $y$ thus
\begin{equation} \label{pfb}
y^s\frac{\p}{\p y^s} \,g_{jk}=0 ,
\end{equation}
and moreover, using Eq.\ (\ref{mod})
\[
\mathscr{L}=\frac{1}{2} g_{ij}(x,y) y^i y^j, \qquad   g_{ij}(x,y) y^i= \frac{\p \mathscr{L}}{\p y^j} .
\]
If  $g$ is positive definite then $\mathscr{L}\ge 0$ and for each $x$, $F(x,y)=\sqrt{2 \mathscr{L}(x,y)}$ provides a norm for $T_xM$. In this case the pseudo-Finsler space is simply called {\em Finsler} space.

\subsection{The Finsler Lagrangian and its non-linear connection}

Non-linear connections which preserve the Finsler Lagrangian have a spray determined by the torsion.

\begin{proposition} \label{mmf}
Every (non-linear) connection $\mathcal{N}$ positive homogeneous of degree one (and with torsion $\tau$) such that
 $\mathscr{L}$ vanishes on the horizontal vectors of $\mathcal{N}$ is such that
\begin{align}
G^i(x,y)&=\mathcal{G}^i(x,y)+\frac{1}{2} g^{ij} \tau_{jk}^l y^k g_{ls} y^s , \label{bhh}\\
\frac{\p^2 N^i_k}{\p y^l\p y^j} \,g_{im} y^m&=\frac{\delta g_{jl}}{\delta x^k}-g_{ji}N^i_{lk}-g_{il}N^i_{jk}=\nabla^{H\mathcal{N}}_{k} g_{jl}. \label{vlc}
\end{align}
where $G^i$ is the spray of $\mathcal{N}$.
Conversely, if the last equation  is satisfied then $\mathscr{L}$ vanishes on the horizontal vectors of $\mathcal{N}$ (thus (\ref{bhh}) holds as well).
\end{proposition}

\begin{proof}
Suppose that $\mathscr{L}$ vanishes on the horizontal vectors of $\mathcal{N}$, then
\begin{equation} \label{aaf}
0=\frac{\p}{\p y^j} \frac{\delta \mathscr{L}}{\delta x^k} = \frac{\p}{\p y^j} \frac{\p \mathscr{L}}{\p x^k}-N^i_k \frac{\p^2 \mathscr{L}}{\p y^j \p y^i}-\frac{\p N^i_k}{\p y^j} \frac{\p \mathscr{L}}{\p y^i}.
\end{equation}
Contracting with $y^k$ and using the definition of torsion in the last term
\[
0=y^k\frac{\p}{\p y^j} \frac{\p \mathscr{L}}{\p x^k}-y^kN^i_k \frac{\p^2 \mathscr{L}}{\p y^j \p y^i}-y^k\frac{\p N^i_j}{\p y^k} \frac{\p \mathscr{L}}{\p y^i}+y^k\tau^i_{jk} \frac{\p \mathscr{L}}{\p y^i}.
\]
Using $y^k N^i_k=2G^i$ and the positive homogeneity of $N^i_j$
\[
0=y^k\frac{\p}{\p y^j} \frac{\p \mathscr{L}}{\p x^k}-2G^i \frac{\p^2 \mathscr{L}}{\p y^j \p y^i}- N^i_j \frac{\p \mathscr{L}}{\p y^i}+y^k\tau^i_{jk} \frac{\p \mathscr{L}}{\p y^i}.
\]
Using $\frac{\delta}{\delta x^j} \mathscr{L}$ to rewrite the penultimate term, and the definition of $\mathcal{G}^i$ given by (\ref{axo}) we arrive at Eq.\ (\ref{bhh}).

Let us differentiate (\ref{aaf}) with respect to $y^l$.
We obtain rearranging terms
\[
\frac{\delta g_{jl}}{\delta x^k}-g_{ji}N^i_{lk}-g_{il}N^i_{jk}=\frac{\p^2 N^i_k}{\p y^l\p y^j}\, g_{im} y^m.
\]
Since all the terms in (\ref{aaf}) are positive homogeneous of degree one it is possible to reobtain (\ref{aaf}) from the last one through multiplication by $y^l$. And from equation (\ref{aaf}) it is possible to obtain ${\delta \mathscr{L}}/{\delta x^k}=0$ through multiplication by $y^j$, which concludes the proof.
\end{proof}

Since
every (non-linear) connection positive homogeneous of degree one is determined  by its spray  and by its torsion as in Eq.\ (\ref{bln}) we have

\begin{corollary} \label{mmg}
There is  one and only one torsionless non-linear connection positive homogeneous of degree one for which the Finsler Lagrangian $\mathscr{L}$ vanishes on the horizontal vectors, namely  (the Berwald non-linear connection) $\mathcal{G}^i_k$.
\end{corollary}

\subsection{Notable Finsler connections}

Every Finsler connection can be locally represented through a triple of coefficients \[
(N^a_i, H^a_{bi}, V^a_{bc})
\]
with appropriate transformation properties under change of coordinates.

Given a Finsler Lagrangian $\mathscr{L}$, we have its spray determined by $\mathcal{G}^i$, and hence the associated canonical Finsler connection (Berwald's) introduced in Def.\ \ref{clp}
\[
(\mathcal{G}^i_k, \mathcal{G}^i_{jk}, 0)\quad \mathcal{G}^i_k=\frac{\p \mathcal{G}^i}{\p y^k}, \quad \mathcal{G}^i_{jk}=\frac{\p^2 \mathcal{G}^i}{\p y^j\p y^k} .
\]
We can determine further Finsler connections imposing the vanishing of the horizontal or vertical covariant derivatives of the metric $g_{ij}$. Through a standard calculation it is easy to see,  imposing Eq.\ (\ref{vlf}) or (\ref{vlg}) and the vanishing of the horizontal and vertical torsions,
\begin{align}
&&\nabla^H g=0, \quad H^i_{jk}=H^{i}_{kj}   \ \Leftrightarrow \ \ H^i_{jk}=\Gamma_{jk}^{ i}:=&\frac{1}{2} g^{is} \left( \frac{\delta}{\delta x^j} \,g_{sk}+\frac{\delta}{\delta x^k} \, g_{sj}-\frac{\delta}{\delta x^s}\, g_{jk}\right) ,\label{mqf}\\
&&\nabla^V g=0, \quad V^i_{jk}=V^{i}_{kj}    \ \Leftrightarrow \ \ V^i_{jk}=C_{jk}^{ i}:=&\frac{1}{2} g^{is} \left( \frac{\p}{\p y^j} \,g_{sk}+\frac{\p}{\p y^k} \, g_{sj}-\frac{\p}{\p y^s}\, g_{jk}\right) \nonumber\\
&{}& {} =&\frac{1}{2} g^{is} \frac{\p}{\p y^j} \,g_{sk} \nonumber
\end{align}
The tensor $C_{ijk}:=g_{is} {C}_{jk}^{s}$ is  symmetric and is called {\em Cartan torsion}.

The most used Finsler connections correspond to the four most noteworthy combinations
\begin{align*}
\textrm{Berwald:}& \quad  (\mathcal{G}^i_k, \mathcal{G}^i_{jk}, 0), \\
\textrm{Cartan:}&  \quad (\mathcal{G}^i_k, \Gamma_{jk}^{i}, {C}_{jk}^{i}), \\
\textrm{ Chern-Rund:}& \quad (\mathcal{G}^i_k,  \Gamma_{jk}^{i}, 0), \\
\textrm{Hashiguchi:}&  \quad (\mathcal{G}^i_k, \mathcal{G}^i_{jk},  {C}_{jk}^{i}).
\end{align*}
The corresponding covariant derivatives will be denoted with $\nabla^{HB}$, $\nabla^{VB}$, $\nabla^{HC}$, $\nabla^{VC}$, the other possibilities being coincident with one of them. By construction
\[
\nabla^{HC} g_{ij}= \nabla^{VC} g_{ij}=0.
\]
An easy computation shows that $\nabla^{HC} g^{ij}=\nabla^{VC} g^{ij}=0$. Moreover,
\[
\nabla^{VB}_i g_{jk}= \frac{\p}{\p y^i} g_{jk}=2C_{ijk}.
\]
The derivative $\nabla^{HB} g_{ij}$ has been calculated in Eq.\ (\ref{vlc}) (see Eq.\ (\ref{nnn})).

Unless otherwise specified we assume that the non-linear connection is Berwald's.

\subsection{Derivatives of the metric and Landsberg tensor}

It is well known that in Riemannian geometry the first derivatives of the metric can be obtained from the Levi-Civita connection coefficients (Christoffel symbols) and hence from the spray. One might ask whether a similar formula holds in Finsler geometry.

Differentiating $g_{si} \mathcal{G}^s$ as given in (\ref{axo}) with respect to $y^j$, we obtain
\begin{equation} \label{bbe}
4C_{sij} \mathcal{G}^s+2g_{si} \mathcal{G}^s_j=y^m\frac{\p^3 \mathscr{L}}{\p x^m \p y^i \p y^j}+\frac{\p^2 \mathscr{L}}{\p x^j \p y^i }-\frac{\p^2 \mathscr{L}}{\p x^i \p y^j}.
\end{equation}
Summing this equation with the same equation with $i$ and $j$ exchanged gives
\begin{equation} \label{use}
y^m\frac{\p}{\p x^m} \,g_{ij}=4C_{sij} \mathcal{G}^s+ g_{si} \mathcal{G}^s_j+  g_{sj}  \mathcal{G}^s_i.
\end{equation}
Differentiating with respect to $y^k$
\begin{align}
\frac{\p}{\p x^k} \,g_{ij}&=-2y^m \frac{\p}{\p x^m} \,C_{ijk}+ 4C_{sijk} \mathcal{G}^s+ 4C_{sij } \mathcal{G}^s_k + 2C_{sik}  \mathcal{G}^s_j \nonumber \\
&\quad \qquad+g_{si} \mathcal{G}^s_{jk}+ 2C_{sjk} \mathcal{G}^s_i+ g_{sj} \mathcal{G}^s_{ik}, \label{xxo}
\end{align}
where
\[
C_{sijk}=\frac{\p}{\p y^k}\, C_{sij}.
\]
This equation gives the first derivatives of the metric tensor with respect to the horizontal variables and can be useful in some computations made in coordinates.

Using positive homogeneity Equation (\ref{xxo}) can be rearranged as follows
\begin{equation}
\nabla^{HB}_i g_{jk}= -2 y^l \nabla^{HB}_l C_{ijk}.
\end{equation}
By Prop.\ \ref{mmf} and Cor.\ \ref{mmg}
\begin{equation} \label{nnn}
\nabla^{HB}_{i} g_{jk}=\mathcal{G}^l_{ijk} g_{lm} y^m.
\end{equation}
These equations are well known \cite{okada82}.
The {\em Landsberg} tensor is the  map $L\colon E \to T^*M\otimes_M T^*M\otimes_M T^*M$, defined by
\begin{equation} \label{mkl}
 \quad L_{ijk}:=-\frac{1}{2}\,\mathcal{G}^l_{ijk} g_{lm} y^m=y^l \nabla^{HB}_l C_{ijk}=-\frac{1}{2} \nabla^{HB}_{i} g_{jk} .
\end{equation}
This definition clarifies that it is symmetric and
\begin{equation} \label{alk}
L_{ijk} y^k=0.
\end{equation}
Subtracting $\nabla^{HC}_i g_{jk}=0$ from Eq.\ (\ref{nnn}) we obtain
\[
-g_{jl} (\mathcal{G}^l_{ki}-\Gamma^l_{ki})-g_{lk} (\mathcal{G}^l_{ji}-\Gamma^l_{ji})=-2 L_{ijk}
\]
since the ${\delta}/{\delta x^i}g_{jk}$ terms cancel out. Now observe that the second term on the left-hand side and the right-hand side are symmetric under exchange of $i$ and $j$, and so also the first term on the left-hand side shares this symmetry. Since it is already symmetric under exchange of $i$ and $k$ it must be totally symmetric and so the two terms on the left-hand side are equal. As a consequence, we obtain another expression of the Landsberg tensor as difference of two connections
\begin{equation} \label{lan}
L_{ijk}=g_{il} (\mathcal{G}^l_{jk}-\Gamma^l_{jk}).
\end{equation}
From here it is immediate that
\[
\Gamma^l_{ki} y^k=\mathcal{G}^l_{ki} y^k=\mathcal{G}^l_{i},
\]
thus all the notable Finsler connections we mentioned are regular. Also from Eq.\ (\ref{lan}) and (\ref{alk}) we have that $y^l\nabla_l^{HC}=y^l\nabla_l^{HB}$ so we can rewrite Eq.\ (\ref{mkl}) as
\begin{equation} \label{llp}
L_{ijk}=y^l \nabla^{HC}_l C_{ijk}.
\end{equation}
Finally, another way of obtaining the Landsberg tensor is as follows. From Eq.\ (\ref{lan}) using $y^j \mathcal{G}^l_{ijk}=0$
\[
y^j \frac{\p}{\p y^i} \,\Gamma^l_{jk}=y^j \mathcal{G}^l_{ijk}-y^j \frac{\p}{\p y^i} \, L^l_{jk}=-\frac{\p}{\p y^i}\,(y^j L^l_{jk})+L^l_{ik}=L^l_{ik}.
\]
In particular if $\Gamma^l_{jk}$ does not depend on $y$, then the Landsberg tensor vanishes and $\mathcal{G}^l_{jk}(=\Gamma^l_{jk})$ does not depend on $y$. This property characterizes the {\em Berwald spaces}, see next sections.


Let us define  the tensors
\[
I_i:= C_{ijk} g^{jk}, \quad J_i:= L_{ijk} g^{jk}, \quad E_{ij}:=\frac{1}{2}\,\mathcal{G}^l_{ijl}
\]
called respectively {\em mean Cartan torsion},  {\em mean Landsberg curvature} and {\em mean  Berwald curvature}. From Eq.\ (\ref{llp})
\begin{equation} \label{mkv}
J_i=y^l \nabla^{HC}_l I_i.
\end{equation}
By Jacobi's formula for the derivative of a determinant
\begin{equation} \label{mnk}
I_i(x,y)= \frac{\p}{\p y^i} \left(\ln \sqrt{\vert \det g_{jk}(x,y)\vert}\right).
\end{equation}
Deicke's theorem  establishes that
a Minkowski norm is Euclidean if and only if $I_i=0$ (a short proof can be obtained joining the first part of the proof in \cite{bao00} with the last part in \cite{mo06}). As a consequence,  every Finsler space for which  $I_i=0$  is Riemannian, namely $g_{jk}$ does not depend on $y$.

\subsubsection{A characterization of Berwald's spaces}
Applying the operator $\p^2/\p y^k \p y^l$ to both sides of Eq.\ (\ref{bbe}) we obtain
\begin{align}
g_{is} \mathcal{G}^s_{jkl}&=y^m\nabla_m^{HB} C_{ijkl}-\nabla_i^{HB} C_{jkl}+ \nabla_j^{HB} C_{ikl}+\nabla_k^{HB} C_{jil}+\nabla_l^{HB} C_{jki}. \label{mlc}
\end{align}
Summing Eq.\ (\ref{mlc}) with the equations obtained exchanging $i$ with $j$, and with $k$, and subtracting that obtained exchanging $i$ with $l$, leads to
\[
g_{is} \mathcal{G}^s_{jkl}+g_{js} \mathcal{G}^s_{ikl}+g_{ks} \mathcal{G}^s_{jil}-g_{ls} \mathcal{G}^s_{jki}=2 y^m\nabla_m^{HB} C_{ijkl}+4\nabla^{HB}_l C_{ijk}.
\]
Expanding  the covariant derivative in terms of the Berwald coefficients it is easy to verify the following identity
\begin{align}
\frac{\p}{\p y^i} L_{jkl}&=\frac{\p}{\p y^i}(y^m \nabla^{HB}_m C_{jkl})=\nabla^{HB}_i C_{jkl}+ y^m\nabla_m^{HB} C_{ijkl}. \label{csc}
\end{align}
Thus, recalling (\ref{mkl}), we arrive at the two equations
\begin{align}
g_{is} \mathcal{G}^s_{jkl}=\frac{\p}{\p y^i}&(y^m \nabla^{HB}_m C_{jkl})-2\nabla_i^{HB} C_{jkl}+ \nabla_j^{HB} C_{ikl}+\nabla_k^{HB} C_{jil}+\nabla_l^{HB} C_{jki}, \label{aal}\\
2\nabla^{HB}_l C_{ijk}&=\frac{\p}{\p y^l}(y^m g_{ms}\mathcal{G}^s_{ijk})+g_{is} \mathcal{G}^s_{jkl}+g_{js} \mathcal{G}^s_{ikl}+g_{ks} \mathcal{G}^s_{jil}-g_{ls} \mathcal{G}^s_{ijk}.
\end{align}
They prove the equivalence  $\mathcal{G}^s_{ikl}=0 \Leftrightarrow \nabla_i^{HB} C_{jkl}=0$, which provides another characterization of Berwald's spaces (for a different, less direct proof see \cite{bao00}).

Another identity to keep in mind is the following. Using Eq.\ (\ref{csc}) we can rewrite (\ref{aal}) removing the covariant derivatives of the Cartan tensor as follows
\begin{equation} \label{jji}
g_{is} \mathcal{G}^s_{\, jkl}=-\frac{\p}{\p y^i} L_{jkl}+\frac{\p}{\p y^j} L_{ikl}+\frac{\p}{\p y^k} L_{jil}+\frac{\p}{\p y^l} L_{jki}-y^m \nabla_m^{HB} C_{ijkl}.
\end{equation}
This equation  tells us that if  $L=0$ (Landsberg space) then  the left-hand side (Berwald curvature) is totally symmetric.

%
%
%

\subsection{Some terminology}
We have all the ingredients to introduce some terminology.
\begin{itemize}
\item[] A {\em pseudo-Finsler space}  is a manifold $M$ over which some Finsler Lagrangian $\mathscr{L}$ has been assigned.
\item[] A local {\em pseudo-Minkowski space}, is a pseudo-Finsler space such that local coordinates exist for which $\mathscr{L}$ is independent of $x$. A similar definition is obtained dropping ``local''. A pseudo-Minkowski space is said Minkowski if the Finsler metric is positive definite (they are finite dimensional Banach spaces with a norm induced by a  strongly convex unit ball).
\item[] A {\em pseudo-Riemannian space} is a pseudo-Finsler space for which  the following equivalent conditions hold
    \begin{itemize}
    \item[(a)] the metric $g$ does not depend on  $y$,
    \item[(b)] the Finsler Lagrangian $\mathscr{L}$ is quadratic in $y$,
    \item[(c)] $C_{ijk}=0$.
    \end{itemize}
\item[] A {\em Berwald space} is a (pseudo-)Finsler space for which  the following equivalent conditions hold
    \begin{itemize}
    \item[(a)] the Berwald  curvature $\mathcal{G}^i_{jkl}$  vanishes,
    \item[(b)] the Berwald connection $\mathcal{G}^i_{jk}$ is independent of $y$,
    \item[(c)] the spray $\mathcal{G}^i$ is $C^2$ on the zero section,
    \item[(d)] $\Gamma^i_{jk}$ is independent of $y$,
    \item[(e)] $\nabla^{HB} C_{ijk}=0$,
    \item[(f)] $R^{VH}_{\textrm{Ber}}=0$.
    \end{itemize}
\item[] A {\em Landsberg space} is a (pseudo-)Finsler space for which the following equivalent conditions hold
\begin{itemize}
    \item[(a)] the Landsberg tensor vanishes: $L_{ijk}=0$ ,
    \item[(b)] the Berwald connection is metric: $\nabla^{HB} g=0$,
    \item[(c)] $y^l\,\nabla^{HB}_l C_{ijk}=0$,
    \item[(d)] $T^{VH}_{\textrm{ver,ChR}}=0$.
\end{itemize}
\item[] A {\em weakly pseudo-Riemannian space} is a pseudo-Finsler space for which $I_{i}=0$. In the positive definite case, this class coincides with that of the Riemannian spaces by Deicke's theorem.\footnote{This definition is new but appropriate given the definitions below. Incidentally, in Finsler geometry is would have been more appropriate to call the Cartan torsion: {\em Riemann curvature}. In this way a space is  pseudo-Riemannian iff the Riemann curvature vanishes.}
\item[] A {\em weakly Berwald space} is a (pseudo-)Finsler space for which $E_{ij}=0$.
\item[] A {\em weakly Landsberg space} is a (pseudo-)Finsler space for which $J_i=0$.
\end{itemize}

All the above equivalences have been already proved, with the exception of (f) in the definition of Berwald space, and (d) in the definition of Landsberg space, which we shall prove in Sections \ref{cuh} and \ref{cui}.

Every pseudo-Riemannian space is Berwald but the converse is not true, consider the pseudo-Minkowski spaces which are not pseudo-Riemannian. Every Berwald space is Landsberg but it is not known whether the converse is true. Every weakly pseudo-Riemannian space is weakly Landsberg by Eq.\ (\ref{mkv}). With Proposition \ref{aes} we shall prove that it is also weakly Berwald.

\begin{theorem}
A pseudo-Finsler space is locally pseudo-Minkowski if and only if $\mathcal{G}^i_{jkl}=0$ and $R^i_{jk}=0$.
\end{theorem}

\begin{proof}
If the space is locally pseudo-Minkowski then there are local coordinates such that $\mathscr{L}$ does not depend on $x$. From Eq.\ (\ref{axo}) we get  $\mathcal{G}^i=0$, thus the coefficients of the non-linear connection $\mathcal{G}^i_j$ vanish, and so $\mathcal{G}^i_{jkl}=0$ and $R^i_{jk}=0$.

Conversely, if $\mathcal{G}^i_{jkl}=0$ then the `Berwald non-linear connection' of coefficients $\mathcal{G}^i_j$ is actually linear. Furthermore, as its curvature $R^i_{jk}$ vanishes by hypothesis, the parallel transport induced by the covariant derivative $D$ is linear and independent of the path followed. As a consequence, starting from a base $\{e_j\}$ of some $T_xM$, we can construct fields denoted in the same way in a neighborhood of $x$. The Berwald non-linear connection has zero torsion and in the new base we have by construction $D_{e_i} e_j=0$, thus  $[e_i,e_j]=-T(e_i,e_j)+D_{e_i} e_j-D_{e_j} e_i=0$ so these fields are holonomic. In conclusion, we can find local coordinates in such a way that $D_{e_i} e_j=0=\mathcal{G}^k_{ji}e_k=0$, thus $\mathcal{G}^k_{i}=\mathcal{G}^k_{ij}y^j=0$. But the Berwald non-linear connection is such that $\frac{\delta}{\delta x^k} \mathscr{L}=(\frac{\p}{\p x^k}-\mathcal{G}^j_k \frac{\p}{\p y^j}) \mathscr{L}=0$ thus $\frac{\p}{\p x^k}\mathscr{L}=0$, which concludes the proof.
\end{proof}

\subsection{Covariant derivative of the volume form} \label{vol}
Let us consider the volume tensor
\[
\mu_{ij\cdots k}=\sqrt{\vert \det g_{lm}(x,y)\vert} \, \epsilon_{ij\cdots k}
\]
where $\epsilon$ is the totally antisymmetric symbol determined by the normalization, $\epsilon_{12,\cdots n}=1$. The horizontal derivative is
\[
\frac{1}{\sqrt{\vert \det g_{lm}\vert}} \nabla^{HC}_s\mu_{ij\cdots k}= [\frac{ \delta}{\delta x^s} \ln \sqrt{\vert \det g_{lm}\vert} ]\,\epsilon_{ij\cdots k}-\Gamma^r_{is} \epsilon_{rj\cdots k}-\Gamma^r_{js} \epsilon_{ir\cdots k}+\cdots
\]
Since the left-hand side is antisymmetric in $ij\cdots k$ we can antisymmetrize the right-hand side with respect to the same indices, obtaining
\[
\nabla^{HC}_s\mu_{ij\cdots k}=\{ \frac{ \delta}{\delta x^s} \ln \sqrt{\vert \det g_{lm}\vert}  -\Gamma^r_{r s}\} \, \mu_{ij\cdots k}
\]
But from the  definition of $\Gamma$ we obtain $g_{is} \Gamma^s_{kj}+g_{k s} \Gamma^s_{ij}=\frac{\delta}{\delta x^j} g_{ik}$ which implies under contraction with $g^{ik}$ and using again Jacobi's formula for the derivative of a determinant
\begin{equation} \label{lll}
\Gamma^l_{li}=\frac{1}{2} g^{jk} \frac{\delta}{\delta x^i} g_{jk}=\frac{1}{2} g^{jk} \frac{\p}{\p x^i} g_{jk}-\frac{1}{2} g^{jk} \mathcal{G}_{i}^l \frac{\p}{\p y^l} g_{jk} =\frac{ \delta}{\delta x^i} \ln \sqrt{\vert \det g_{lm}\vert} .
\end{equation}
Thus  \[\nabla^{HC} \mu=0.\]
A completely analogous computation shows that $\nabla^{VC} \mu=0$.
Repeating the same steps for the Berwald connection does not lead to the same result since from Eq.\ (\ref{lan}) and (\ref{lll})
\[
\mathcal{G}^l_{li}=\frac{ \delta}{\delta x^i} \ln \sqrt{\vert \det g_{lm}\vert}+J_i.
\]
Thus we arrive at
\[
\nabla^{HB}_s\mu_{ij\cdots k}=-J_s  \, \mu_{ij\cdots k}.
\]
Similarly, we obtain
\[
\nabla^{VB}_s\mu_{ij\cdots k}=I_s  \, \mu_{ij\cdots k}.
\]

As a consequence:

\begin{theorem}
The weakly (pseudo-)Riemannian spaces are those \mbox{(pseudo-)} Finsler spaces for which the volume form induced by the metric is covariantly constant with respect to the vertical Berwald covariant derivative. By Eq. (\ref{mnk}) they are also those spaces for which the volume form is independent of $y$.

The weakly (pseudo-)Landsberg spaces are those (pseudo-)Finsler spaces for which the volume form induced by the metric is  covariantly constant with respect to the horizontal Berwald covariant derivative.
\end{theorem}

\section{Some useful identities}

In this section we establish some useful identities for  the non-linear connection and the linear connections of Finsler geometry.

\subsection{The Bianchi identities for the non-linear connection}

Let us start with the second Bianchi identities since they do not require a soldering form.

\subsubsection{Second Bianchi identities}
The second Bianchi identities of a non-linear connection are
\begin{equation} \label{bin}
[\mathcal{N},R]=0,
\end{equation}
where $[,]$ is the Fr\"olicher-Nijenhuis bracket. They read in components
\[
0=(\mathcal{N}^\nu_\alpha \p_\nu R^\mu_{\beta \gamma} -R^\nu_{\alpha \beta} \p_\nu \mathcal{N}^\mu_{ \gamma} - \mathcal{N}^\mu_{ \nu} \p_\alpha R^\nu_{\beta \gamma}+2R^\mu_{\alpha \nu} \p_\beta \mathcal{N}^\nu_{ \gamma}) \,\frac{\p}{\p z^\mu}\otimes \dd z^\alpha \wedge \dd z^\beta \wedge \dd z^\gamma
\]
The third and fourth terms vanish identically, and so $\alpha, \beta, \gamma$ must be indices of base type and $\mu$ of fiber type.
\begin{equation} \label{bim}
0= (\p_i R^a_{jk}-N^b_i \frac{\p}{\p y^b} R^a_{jk} +R^b_{i j}  \frac{\p}{\p y^b}N^a_{k}) \,\frac{\p}{\p y^a}\otimes \dd x^i \wedge \dd x^j\wedge \dd x^k.
\end{equation}
Introduced the canonical Finsler connection $(N^a_i, N^a_{bi},0)$ of the non-linear connection, and its horizontal covariant derivative $\nabla^{HN}$, the previous equation  reads
\[
0= ( \nabla_i^{HN} R^a_{jk} +R^a_{b k}  N^b_{ji}+R^a_{jb} N^b_{ki}) \,\frac{\p}{\p y^a}\otimes \dd x^i \wedge \dd x^j\wedge \dd x^k
\]
\[
0= ( \nabla_i^{HN} R^a_{jk} +R^a_{b i}  \tau^b_{jk}) \,\frac{\p}{\p y^a}\otimes \dd x^i \wedge \dd x^j\wedge \dd x^k
\]
thus
\[
\nabla_{[i}^{HN} R^a_{jk]}+R^a_{b [i}  \tau^b_{jk]}=0 ,
\]
where $[\,]$ denotes antisymmetrization of the indices $i,j,k$.
%
Observe that Eq.\ (\ref{bin})-(\ref{bim}) depend just on the connection, and do not depend on a soldering form. Nevertheless, $\nabla^{HN}$ depends on the soldering form thus expressing the Bianchi identity in terms of this covariant derivative reintroduces the soldering form through its torsion.

If $N^a_i$ is the Berwald non-linear connection $\mathcal{G}^a_i$, then we can replace in the previous formula $\nabla^{HN}$ for $\nabla^{HB}$, and set $\tau=0$,
\begin{equation} \label{fpo}
\nabla_i^{HB} R^a_{jk}+\nabla_j^{HB} R^a_{ki}+\nabla_k^{HB} R^a_{ij}=0.
\end{equation}
If $N^a_i$ is the Berwald non-linear connection $\mathcal{G}^a_i$, but we wish to use $\nabla^{HC}$ then we arrive at
\begin{equation} \label{fpp}
\nabla_i^{HC} R^a_{jk}+\nabla_j^{HC} R^a_{ki}+\nabla_k^{HC} R^a_{ij}+R^l_{jk} L^a_{li}+R^l_{ki} L^a_{lj}+ R^l_{ij} L^a_{lk}=0.
\end{equation}

\subsubsection{First Bianchi identities}
Let  $e\colon E\to T^*B\otimes_E VE$ be a soldering form.
The first Bianchi identity is
\[
[\mathcal{N},\tau]+[R,e]=0.
\]
Using again the canonical soldering form  ($e^a_i=\delta^a_i$)
\[
[\p_i \tau^a_{jk}- N^b_i \frac{\p}{\p y^b} \tau^a_{jk}+\tau^b_{jk} N^a_{bi}-e^b_i\frac{\p}{\p y^b} R^a_{j k}] \,\frac{\p}{\p y^a}\otimes \dd x^i \wedge \dd x^j\wedge \dd x^k=0
\]
that is
\begin{equation} \label{mke}
\nabla^{HN}_{[i} \tau^l_{jk]}+\tau^l_{s[i} \tau^s_{jk]}-\frac{\p}{\p y^{[i}} R^l_{jk]}=0,
\end{equation}
where $[\,]$ denotes antisymmetrization. In Sect.\ \ref{mog} we shall see that it can be recasted in a more familiar form.
%

\subsection{Curvature of the linear connection} \label{cuh}
The linear connection $\nabla$ allows us to introduce a curvature
\[
R^\nabla(\check X,\check Y)\colon VE\to VE,
\]
in the usual way. Namely, let $\check X,\check Y\colon E\to TE$ and $\tilde{Z}\colon E\to VE$ (which is equivalent as assigning  fibered morphism $Z:E\to E$, where $\tilde{Z}$ is the {\em vertical lift} of $Z$, also denoted $Z^V$), and define
\begin{equation} \label{lld}
R^\nabla(\check X,\check Y) \tilde{Z}= \nabla_{\check X}  \nabla_{\check Y} \tilde{Z}- \nabla_{\check Y}  \nabla_{\check X}\tilde{Z}-\nabla_{[\check X,\check Y]} \tilde{Z}.
\end{equation}
However, $\check X,\check Y\colon E\to TE$ can be chosen horizonal or vertical leading to a splitting analogous to that obtained for the connection.

%

\subsubsection{The HH-curvature} \label{mog}
Let $Z\colon E\to E$ be a fibered morphisms and let $X,Y\colon M\to TM$ be sections, we define
\begin{align*}
R^{HH}(X,Y)Z &:=R^\nabla(\mathcal{N}(X),\mathcal{N}(Y)) \tilde{Z}\\
&=\nabla_{\mathcal{N}(X)} \nabla_{\mathcal{N}(Y)} \tilde{Z}-\nabla_{\mathcal{N}(Y)}  \nabla_{\mathcal{N}(X)}  \tilde{Z}-\nabla_{[\mathcal{N}(X),\mathcal{N}(Y)]} \tilde{Z}.
\end{align*}
Now we use an equivalent definition  for the curvature of a non-linear connection\footnote{This is Eq.\ (1) of \cite{modugno91} with our choice of coefficients for the curvature. It can be obtained from Eq.\ (\ref{nns}) with $\mathcal{N}(X)=X^i(x)\,\delta/\delta x^i$, $\mathcal{N}(Y)=Y^j(x) \,\delta/\delta x^j$.}
\begin{equation} \label{dom}
[\mathcal{N}(X),\mathcal{N}(Y)]=\mathcal{N}([X,Y])-R(X,Y) ,
\end{equation}
thus
\begin{align*}
R^{HH}(X,Y) Z&=[\nabla^H_X \nabla^H_YZ-\nabla^H_Y \nabla^H_XZ-\nabla^H_{[X,Y]} Z]+\nabla^V_{R(X,Y)} Z.
\end{align*}
 In components
\[
R^{HH}{}^i_{\, jkl} Z^j=(\nabla^H_k \nabla^H_l Z- \nabla^H_l \nabla^H_k Z)^i+ (\nabla^V_m Z)^i R^m_{kl}.
\]
All the notable Finsler connections satisfy $V^a_{bc}(x,y) y^b=0$, thus $(\nabla^V_c L)^a=\delta^a_c$ where $L$ is the Liouville vector field, and hence
\begin{equation} \label{mco}
R^{HH}{}^i_{\, jkl} y^j=R^{i}_{\, kl}.
\end{equation}
Equivalently,
\[
R^{HH}(X,Y)L=R(X,Y).
\]
In particular this contraction of the HH-curvature depends only on the curvature $R$ of the induced non-linear connection. More generally, we get
\begin{equation} \label{vll}
R^{HH}{}^i_{\, jkl}=\frac{\delta}{\delta x^k} \,H^i_{jl}-\frac{\delta}{\delta x^l} \, H^i_{jk}+H^i_{mk} H^m_{jl}-H^i_{ml} H^m_{jk}+R^m_{kl} V^i_{jm} .
\end{equation}
In the Berwald case, using Eq.\ (\ref{nnd}) and (\ref{cur}) we obtain an improvement over Eq.\  (\ref{mco}). With obvious meaning of the notation
\begin{equation} \label{msp}
\frac{\p}{\p y^j} R^i_{kl}=R^{HH}_{\textrm{Ber}}{}^i_{\, jkl}= \frac{\delta}{\delta x^k}\, \mathcal{G}^i_{jl}-\frac{\delta}{\delta x^l} \, \mathcal{G}^i_{jk}+\mathcal{G}^i_{mk} \mathcal{G}^m_{jl}-\mathcal{G}^i_{ml} \mathcal{G}^m_{jk}.
\end{equation}
Using Eq.\ (\ref{lan}) we obtain
\begin{equation} \label{mmx}
R^{HH}_{\textrm{Ber}}{}^i_{\, jkl}=R^{HH}_{\textrm{ChR}}{}^i_{\, jkl}+\nabla^{HC}_k L^i_{jl}-\nabla^{HC}_l L^i_{jk}+g^{is}(L_{skm} L_{jln}-L_{slm} L_{jkn}) g^{mn} ,
\end{equation}
which, lowering the first index, can be rewritten
\begin{equation}
R^{HH}_{\textrm{Ber}}{}_{ijkl}=R^{HH}_{\textrm{ChR}}{}_{ijkl}+\nabla^{HC}_k L_{ijl}-\nabla^{HC}_l L_{ijk}+(L_{ikm} L_{jln}-L_{ilm} L_{jkn}) g^{mn} .
\end{equation}
Using Eq. (\ref{vll}) it is easy to see that (compare \cite[Eq.\ 1.34 Chap.\ IV]{rund59})
\begin{equation} \label{mmz}
R^{HH}_{\textrm{Car}}{}^i_{\, jkl}=R^{HH}_{\textrm{ChR}}{}^i_{\, jkl}+R^m_{kl} C^i_{jm} .
\end{equation}

Thanks to these results, the first Bianchi identity for the non-linear connection, Eq.\ (\ref{mke}), takes the more familiar form
\begin{equation} \label{llf}
R^{HH}_{\textrm{Ber}}{}^i_{\, [jkl]}=R^{HH}_{\textrm{ChR}}{}^i_{\, [jkl]}=0.
\end{equation}
They can be identified with the first Bianchi identities for the HH-curvature.
For  the Cartan Finsler connection these identities take the form
\begin{equation} \label{llg}
R^{HH}_{\textrm{Car}}{}^i_{\, [jkl]}=R^m_{[jk} C^i_{l]m} .
\end{equation}

\subsubsection{The VH-curvature}
Let $Z\colon E\to E$ be a fibered morphisms and let $X,Y\colon M\to TM$ be sections. The mixed vertical-horizontal curvature is
\begin{align*}
R^{VH}(X,Y) Z&=-R^{HV}(Y,X) Z:=\nabla_{\tilde X}  \nabla_{\mathcal{N}(Y)}  \tilde Z-\nabla_{\mathcal{N}(Y)} \nabla_{\tilde X} \tilde Z-\nabla_{[\tilde X, \,\mathcal{N}(Y)]} \tilde Z,
\end{align*}
thus
\[
R^{VH} {}^i_{\, jkl}=-R^{HV} {}^i_{\, jlk}= -\frac{\delta}{\delta x^l} V^i_{jk}+\frac{\p}{\p y^k} H^i_{jl}-H_{ml}^i V^m_{jk}+V^i_{mk} H^m_{jl}+V^i_{j m} N^m_{kl}.
\]
For the Berwald Finsler connection
\[
R^{VH}_{Ber} {}^i_{\, jkl}= \mathcal{G}^i_{jkl}.
\]
For the Cartan Finsler connection
\begin{align}
R^{VH}_{\textrm{Car}} {}^i_{\, jkl}&= -\nabla^{HC}_l C^i_{jk}+ \frac{\p}{\p y^k} \Gamma^{i}_{jl}+C^i_{jm} L^m_{kl} \nonumber \\
&= \mathcal{G}^i_{jkl}-\frac{\p}{\p y^k} L^{i}_{jl}+C^i_{jm} L^m_{kl}-\nabla^{HC}_l C^i_{jk}. \label{dju}
\end{align}
For the Chern-Rund Finsler connection (a tensor proportional to this one is often denoted with $P$, see \cite{bao00})
\begin{equation}
R^{VH}_{\textrm{ChR}} {}^i_{\, jkl}= \frac{\p}{\p y^k} \Gamma^{i}_{jl}=\mathcal{G}^i_{jkl}-\frac{\p}{\p y^k} L^{i}_{jl}.
\end{equation}
For the Hashiguchi Finsler connection
\[
R^{VH}_{\textrm{Has}} {}^i_{\, jkl}= \mathcal{G}^i_{jkl}-\nabla^{HB}_l C^i_{jk}.
\]
It is interesting to notice that
\begin{equation} \label{mlg}
R^{VH}_{\textrm{ChR}} {}^m_{\, mkl}=\nabla^{HB}_l I_k,
\end{equation}
as can be easily obtained using Eq.\ (\ref{lll}).

\subsubsection{The VV-curvature}
Analogously, given sections $X,Y:M\to TM$, and a fibered morphism $Z\colon E\to E$, let us define
\begin{align*}
R^{VV}(X,Y) Z&=\nabla_{\tilde X} \nabla_{\tilde Y} \tilde Z-\nabla_{\tilde Y}  \nabla_{\tilde X}  \tilde Z-\nabla_{[\tilde X, \tilde Y]} \tilde Z,
\end{align*}
where $\tilde{X}\colon E\to VE=E\times_M E$ is such that its second component is $X$, and analogously for $\tilde{Y}$ (this is also called {\em vertical lift} and denoted $X^V$).
Since the vertical distribution is integrable, its integral being the fiber, we obtain
\begin{align*}
R^{VV}(X,Y) Z&=\nabla_{X}^V \nabla^V_{ Y} Z-\nabla^V_{ Y}  \nabla^V_{X}  Z-\nabla^V_{[ X, Y]} Z,
\end{align*}
and hence in components
\begin{equation}
R^{VV} {}^i_{\ jkl}= \frac{\p}{\p y^k} V^i_{jl}-\frac{\p}{\p y^l} V^i_{jk} +V^i_{m k} V^m_{jl}-V^i_{ml} V^m_{jk}.
\end{equation}
There are two interesting cases, namely the trivial one $V^a_{bc}=0$ (Berwald and Chern-Rund) and the case $V^a_{bc}=C^a_{bc}$, which holds for the Cartan and Hashiguchi Finsler connections. For the latter case, let us use $\frac{\p}{\p y^k} \, g^{is}=-2g^{ip} C_{pqk} g^{qs}$ to show first that
\[
\frac{\p}{\p y^k}\, C^i_{jl}= -2 C^i_{mk} C^m_{jl} +g^{im} C_{mjlk},
\]
and hence that (see also Rund \cite[Chap.\ IV]{rund59} and \cite{kikuchi62})
\begin{equation} \label{kik}
R^{VV} {}^i_{jkl}= C^i_{ml} C^m_{jk}-C^i_{m k} C^m_{jl}.
\end{equation}
The VV-Ricci tensor is
\[
R^{VV}_{jl}=R^{VV} {}^s_{jsl}= C^s_{ml} C^m_{js}-I_m C^m_{jl}.
\]
An interesting theorem by Matsumoto \cite{matsumoto77} establishes that in a pseudo-Finsler manifold which is 4-dimensional $R^{VV}_{jl}=0$ if and only if $R^{VV} {}^i_{jkl}=0$. Thus in the positive definite case, by Brickell theorem \cite{brickell67}, $C_{ijk}=0$, namely the Finsler space is Riemannian.

\subsection{Torsion of the linear connection} \label{cui}

The torsion of $\nabla$ can be introduced through
\begin{equation} \label{too}
T^{(\nabla,\sigma)}(\check X,\check Y)= \nabla_{\check X} \,\sigma(\check Y)-\nabla_{\check Y}\, \sigma(\check X)-\sigma([\check X,\check Y]) ,
\end{equation}
only once  a  soldering form $\sigma\colon TE\to VE$ has been given.\footnote{Of course Eq.\ (\ref{tor}) for a linear connection on the bundle $\pi_E:\tilde{E}\to E$ would give the same definition.} We shall only consider surjective soldering forms. There are two natural choices for the soldering form.

\subsubsection{`Horizontal' soldering form}
It is natural to introduce the soldering form $\sigma_h$ through $\sigma_h(\check X)=(T\pi_M(\check X))^V$ so that, given $X:M\to TM$, $\sigma_h(\mathcal{N} (X))=\tilde X$. Thus we define
\[
T_{\textrm{hor}}^{HH}(X,Y)= T^{(\nabla,\sigma_h)}(\mathcal{N}(X),\mathcal{N}(Y))=\nabla^H_X Y-\nabla^H_Y X-[X,Y],
\]
where we used Eq.\ (\ref{dom}) and $\sigma_h(R(X,Y))=0$. Its components in a holonomic base are
\[
T_{\textrm{hor}}^{HH} {}^k_{ij}=H_{ji}^k-H^k_{ji},
\]
 so this torsion vanishes for all the notable Finsler connections.

If the first lifted vector is taken vertical we obtain another torsion
\[
T_{\textrm{hor}}^{VH}(X,Y)= T^{(\nabla,\sigma_h)}(\tilde{X},\mathcal{N}(Y))=\nabla^V_X Y-\sigma_h[\tilde X,\mathcal{N}(Y)]=V^k_{ji} Y^j X^i \frac{\p}{\p y^k},
\]
that is
\[
T_{\textrm{hor}}^{VH} {}^k_{ij}=V^k_{ji},
\]
which vanishes for Berwal's and Chern-Rund's. This equation shows that it is indeed appropriate to call the tensor $C^i_{jk}$ the {\em Cartan torsion}, as it is a torsion of the Cartan connection. The vertical-vertical version vanishes identically.

\subsubsection{`Vertical' soldering form}
Let us define $\sigma_v={\nu}$, so that $\sigma_v(\tilde{X})=X$, $\sigma(\mathcal{N}(X))=0$.
\[
T_{\textrm{ver}}^{VV}(X,Y)= T^{(\nabla,\sigma_v)}(\tilde X,\tilde Y)=\nabla^V_X Y-\nabla^V_Y X-[X,Y] .
\]
Its components are
\[
T_{\textrm{ver}}^{VV} {}^k_{ij}=V_{ji}^k-V_{ij}^k,
\]
so this torsion vanishes for all the notable Finsler connections.

If the second lifted vector is taken horizontal we obtain
\[
T_{\textrm{ver}}^{VH}(X,Y)=T^{(\nabla,\sigma_v)}(\tilde X,\mathcal{N}(Y))=-\nabla^H_Y X-{\nu}([\tilde{X},\mathcal{N}(Y)])=-\nabla^H_Y X+\nabla^{HN}_Y X,
\]
which has components
\[
T_{\textrm{ver}}^{VH} {}^k_{ij}=-H^k_{ij}+N^k_{ij}.
 \]
 Thus this torsion vanishes  for the Berwald and Hashiguchi connections, while it coincides with the Landsberg tensor for the Cartan and Chern-Rund connections. Let us evaluate $R^{VH}(X,Y) Z$ where $Z=L$, with $L$ the Liouville vector field. We have mentioned that for all the notable Finsler connections  $\nabla_{v} L=\tilde \nu(v)$, which is equivalent to $V^a_{ij}y^i=0$. As a consequence,
\[
R^{VH}(X,Y)L=T^{VH}_{ver}(Y,X),
\]
or in components
\[
R^{VH} {}^i_{jkl} y^j=T^{VH}_{ver}{}^i_{kl}.
\]

The horizontal-horizontal version  gives
\[
T_{\textrm{ver}}^{HH}(X,Y)=T^{(\nabla,\sigma_v)}(\mathcal{N}(X),\mathcal{N}(Y))
=-{\nu}([\mathcal{N}(X),\mathcal{N}(Y)])=R(X,Y),
\]
where we used once again Eq.\ (\ref{dom}). Thus this torsion coincides with the non-linear curvature.

\subsubsection{Characterization of  Finsler connections}

From the previous results it is easy to characterize the notable  Finsler connections from their torsion properties.
\begin{proposition}
Let us consider the family of regular Finsler connections whose induced non-linear connection is Berwald's. The notable Finsler connections satisfy
\begin{equation} \label{vvf}
T^{HH}_{hor}=T^{VV}_{ver}=0.
\end{equation}
Under this equality a linear Finsler connection coincides with the Berwald connection if and only if $T^{VH}_{ver}=0$ and $T^{VH}_{hor}=0$.

The Cartan connection is the only Finsler connection which satisfies (\ref{vvf}) and is metric compatible.

The Chern-Rund connection is the only Finsler connection which satisfies $T^{HH}_{hor}=T^{VH}_{hor}=0$ (hence $T^{VV}_{ver}=0$) and which is horizontally metric compatible $\nabla^H g=0$.

The Hashiguchi Finsler connection is the only Finsler connection which satisfies  $T^{VV}_{ver}=T^{VH}_{ver}=0$ (hence $T^{HH}_{hor}=0$) and which is vertically metric compatible $\nabla^V g=0$.

\end{proposition}

As for every Finsler connection, the  Berwald connection is not torsionless since $T^{HH}_{ver}=R$. Still, it is as torsionless as a Finsler connection can be.
Some authors use only the horizontal soldering form and for this reason they claim that the Berwald connection is torsionless \cite{akbarzadeh88}.

\subsection{Curvature symmetries induced by metric compatibility}

In Riemannian geometry it is well known that the curvature of a {\em Levi-Civita} connection  shares an additional symmetry due the compatibility with the metric. The same happens in Finsler geometry.

Let $\tilde Z_1,\tilde Z_2 \colon E\to VE$, $ \check X, \check Y\colon E\to TE$ be vector fields. Let $\nabla$ be the Cartan Finsler connection, so that $\nabla g=0$ (as this equation holds over the horizontal and vertical vectors). We have
\begin{align*}
g(\tilde Z_1, R^\nabla(\check X,\check Y) \tilde{Z}_2)&= g(\tilde Z_1, \nabla_{\check X}  \nabla_{\check Y} \tilde{Z}_2- \nabla_{\check Y}  \nabla_{\check X}\tilde{Z}_2-\nabla_{[\check X,\check Y]} \tilde{Z}_2)\\
&= \nabla_{\check X} [g(\tilde Z_1,   \nabla_{\check Y} \tilde{Z}_2)]-g(\nabla_{\check X} \tilde Z_1,   \nabla_{\check Y} \tilde{Z}_2) - \nabla_{\check Y}  [g(\tilde Z_1,  \nabla_{\check X}\tilde{Z}_2)]\\
&+g(\nabla_{\check Y} \tilde Z_1,  \nabla_{\check X}\tilde{Z}_2)- \nabla_{[\check X,\check Y]} g(\tilde Z_1,\tilde Z_2)+g(\nabla_{[\check X,\check Y]} \tilde Z_1, \tilde Z_2)\\
&=\cdots=g(R^\nabla(\check Y,\check X) \tilde Z_1, \tilde{Z}_2)=-g(R^\nabla(\check X,\check Y) \tilde Z_1, \tilde{Z}_2) .
\end{align*}
Evaluating this equation with $\check X$, $\check Y$ in the combinations, horizontal-horizontal, vertical-horizontal and vertical-vertical, we get
\begin{align}
R^{HH}_{Car}\,{}_{ijkl}&=-R^{HH}_{Car}\,{}_{jikl}, \label{fir}\\
R^{VH}_{Car}\,{}_{ijkl}&=-R^{VH}_{Car}\,{}_{jikl}, \label{sec}\\
R^{VV}_{Car}\,{}_{ijkl}&=-R^{VV}_{Car}\,{}_{jikl} . \label{thi}
\end{align}

%
Here, as with the other tensors, we lowered the upper index to the left $R_{ijkl}=g_{is} R^{i}_{\, jkl}$ (observe that in order to avoid ambiguities we do not lower, in any circumstance, the upper index of the curvature $R^i_{jk}$ of the non-linear connection).

The VV-symmetry is not particularly interesting as it is identically satisfied due to Eq.\ (\ref{kik}).


\subsubsection{The HH-symmetry}


Let us  show a different approach to  (\ref{fir}) and related symmetries.
Denoting $\frac{\delta}{\delta x^\alpha} f$ with $f_{:\,\alpha}$, and using Eq.\ (\ref{mqf}) the $HH$-Chern-Rund curvature can be rewritten
\begin{align*}
R^{HH}_{\textrm{ChR}} {}_{ijkl}&=\Gamma_{ijl:\,k}-\Gamma_{ijk:\,l}-\Gamma_{mik} \Gamma^m_{jl}+\Gamma_{mil}\Gamma^m_{jk}\\
&\!\!\!\!\!\!\!\!\!\!\!=\frac{1}{2}(g_{ij:l:k}-g_{ij:k:l}+g_{jk:i:l}+g_{li:j:k}-g_{jl:i:k}-g_{ki:j:l}) -\Gamma_{mik} \Gamma^m_{jl}+\Gamma_{mil}\Gamma^m_{jk}
\end{align*}
where we lowered the indices of the connection coefficients to the left.
From here, using Eq.\ (\ref{nns}) it is easy to find the symmetry for the Chern-Rund curvature. For instance
\[
R^{HH}_{\textrm{ChR}} {}_{ijkl}-R^{HH}_{\textrm{ChR}} {}_{klij}=R^m_{ki} C_{mjl}-R^m_{kj} C_{mli}-R^m_{li} C_{mjk}+R^m_{lj} C_{mki}-R^m_{kl} C_{mji}+R^m_{ij} C_{mkl}.
\]
 In terms of the Cartan curvature it reads
\begin{equation} \label{don}
R^{HH}_{\textrm{Car}} {}_{ijkl}-R^{HH}_{\textrm{Car}} {}_{klij}=R^m_{ki} C_{mjl}-R^m_{kj} C_{mli}-R^m_{li} C_{mjk}+R^m_{lj} C_{mki}.
\end{equation}
Using Eqs. (\ref{mmx})-(\ref{mmz}) we obtain from (\ref{fir})
\begin{align}
R^{HH}_{\textrm{Ber}}\,{}_{ijkl}+R^{HH}_{\textrm{Ber}}\,{}_{jikl}&= -2R^m_{kl} C_{ijm} +2(\nabla^{HC}_k L_{ijl}-\nabla^{HC}_{l} L_{ijk}), \\
R^{HH}_{\textrm{ChR}}\,{}_{ijkl}+R^{HH}_{\textrm{ChR}}\,{}_{jikl}&=-2R^m_{kl} C_{ijm} .  \label{fyg}
\end{align}
{}\\
\indent {\bf Contracted HH-symmetry}.\\

\noindent  Contracting the last two equations
\begin{align}
R^{HH}_{\textrm{Ber}}\,{}^i_{\, ikl}&= -R^m_{kl}I_m+\nabla^{HC}_k J_l-\nabla^{HC}_l J_k, \\
R^{HH}_{\textrm{ChR}}\,{}^i_{\, ikl} &=-R^m_{kl}I_m, \label{bka}
\end{align}
which using the first Bianchi identities (\ref{llf})-(\ref{llg}) leads upon contraction of $i$ and $k$ to the following symmetry property of the HH-Ricci tensor
\begin{align}
R^{HH}_{\textrm{Ber}}{}^k_{\, jkl}-R^{HH}_{\textrm{Ber}}{}^k_{\, lkj}&=R^m_{lj}I_m +\nabla^{HC}_j J_l-\nabla^{HC}_l J_j, \label{ohl}\\
R^{HH}_{\textrm{ChR}}{}^k_{\, jkl}-R^{HH}_{\textrm{ChR}}{}^k_{\, lkj}&=R^m_{lj}I_m, \label{ohh}\\
R^{HH}_{\textrm{Car}}{}^k_{\, jkl}-R^{HH}_{\textrm{Car}}{}^k_{\, lkj}&=R^m_{lj}I_m+R^m_{js} C^s_{lm}-R^m_{ls} C^s_{jm}.
\end{align}
%
%
%
%

\subsubsection{The VH-symmetry}

Let us consider the symmetry
\begin{equation}
R^{VH}_{\textrm{Car}}\,{}_{ijkl}=-R^{VH}_{\textrm{Car}}\,{}_{jikl},
\end{equation}
which using Eq.\ (\ref{dju}) and lowering the upper index of $\mathcal{G}^i_{\, jkl}$ to the left, becomes
\[
\mathcal{G}_{ijkl}+\mathcal{G}_{jikl}-2\nabla^{HC}_l C_{ijk}-2\frac{\p}{\p y^k} L_{ijl}+2C_{imk} L^m_{jl}+2C_{ijm} L^m_{kl}+2C_{mjk} L^m_{il}=0.
\]
We can rewrite it as
\begin{equation} \label{mkw}
\mathcal{G}_{ijkl}+\mathcal{G}_{jikl}-2\nabla^{HB}_l C_{ijk}-2\frac{\p}{\p y^k} L_{ijl}=0.
\end{equation}

An interesting equation is obtained as follows. Let us symmetrize (\ref{mkw}) in the indices $k,l$. Subtract to the found equation  the same equation with $k$ and $j$ exchanged and  sum to it the same equation with $k$ and $i$ exchanged and finally subtract Eq.\ (\ref{aal}). This computation gives
\[
\mathcal{G}_{ijkl}=\nabla^{HB}_i C_{jkl}+\frac{\p}{\p y^k} L_{ijl}+\frac{\p}{\p y^j} L_{ilk}+\frac{\p}{\p y^l} L_{ijk}-2\frac{\p}{\p y^i} L_{jkl},
\]
which seems particularly interesting especially for the study of Landsberg or weakly-Landsberg spaces (it could have been obtained from Eqs. (\ref{csc}) and (\ref{jji})).
{}\\

\indent {\bf Contracted VH-symmetry}.\\

\noindent Contracting (\ref{mkw}) with $g^{ij}$ we obtain the very interesting equality (use $\nabla^{HB}_l g^{ij}=2L^{ij}_{\ l}$ and $\frac{\p}{\p y^k} \, g^{ij}=-2C^{ij}_{\ k}$)
\begin{equation} \label{mfg}
\mathcal{G}^i_{\,\, ilk}=\mathcal{G}^i_{\,\, ikl}=\nabla^{HB}_l I_k+\frac{\p}{\p y^k} J_l,
\end{equation}
which contracted once again with $g^{kl}$ gives
\begin{equation}
\mathcal{G}^i_{\, \, i}{}^k_{\,\, k}=\nabla^{HB}_k I^k+\frac{\p}{\p y^k} J^k.
\end{equation}
An important consequence of Eqs.\ (\ref{mkv}), (\ref{ohl})-(\ref{ohh}) and (\ref{mfg}) which does not seem to have been noticed before is

\begin{theorem} \label{aes}
Every weakly (pseudo-)Riemannian space  has symmetric Chern-Rund and  Berwald HH-Ricci tensors. Moreover, these  spaces are weakly Berwald.
\end{theorem}

\begin{remark}
This result is important for the Finslerian generalization of general relativity. Many authors tried to impose equations analogous to Einstein's but faced the problem of non-symmetric Ricci tensors and non-symmetric stress energy tensors. The previous result suggests to impose, aside with some tensorial generalization of the Einstein's equations, (there are a few possibilities, e.g.:)
\begin{equation} \label{djf}
R^{HH}_{\textrm{ChR}} {}_{ij}-\frac{1}{2}\, R^{HH}_{\textrm{ChR}} \,g_{ij}+\Lambda g_{ij}=8\pi T_{ij},
\end{equation}
 the equation
\begin{equation} \label{lao}
I_j=0,
\end{equation}
namely to work with {\em weakly Lorentzian-Finsler spaces}. This fact does not seem to have been previously recognized. In fact, most authors working in Fislerian gravity theories actually did their calculations in the positive definite case with the idea of obtaining the physical Lorentz-Finsler equations through a straightforward tensorial generalization.
%
However, by Deicke's theorem, the previous equation is too restrictive in the positive definite case so  one should work directly with Lorentz-Finsler spaces.

Equation (\ref{djf}) is not completely satisfactory because it does not imply the stress-energy conservation. This problem can be solved provided the Finsler space is Landsberg:
\begin{equation} \label{lop}
L_{jkl}=0.
\end{equation}
Observe that the contracted HHH-second Bianchi identity (\ref{sop}) reads, contracting $k$ with $l$ and $s$ with $j$ (see Sect.\ \ref{xib}, since we assume that the space is Landsberg's  the Chern-Rund and Berwald connections coincide and so we can omit the reference to the  connection)
\begin{equation} \label{kkg}
\nabla^H_j (R^{HH} {}^{m j}_{\ \ \, m i} +R^{HH} {}^{jm}_{\ \ \,\, im}-\delta^j_i R^{HH} {}^{m l}_{\ \ \, m l})=\mathcal{G}^{l s}_{\ m i} R^m_{ls}+\mathcal{G}^{ls}_{\ m l} R^{m}_{si} +\mathcal{G}^{ls}_{\ ms} R^m_{il}=0.
\end{equation}
The last equality follows because by Eq.\ (\ref{jji}) for a Landsberg space $\mathcal{G}_{ijkl}$ is totally symmetric and because of Theorem \ref{aes}.
We define the symmetric tensor
\[
G_{ij}:=\frac{1}{2}(R^{HH} {}^{ \ m}_{i \ \,\, jm} +R^{HH} {}^{\ m}_{j \ \,\, im}-g_{ij} R^{HH} {}^{m l}_{\ \ \, m l} ),
\]
and write the field equations as
\begin{equation} \label{max}
G_{ij}+\Lambda g_{ij}=8\pi T_{ij}.
\end{equation}
From Eq.\ (\ref{bka}) we obtain $R^{HH}{}^m_{\ m kl}=-R^m_{kl} I_m=0$, and from Eq.\ (\ref{fyg}) contracting $j$ with $l$ and using the fact that by $I_j=0$ Ricci is symmetric, and using Eq.\ (\ref{fyg}), $R^{HH} {}^m_{\ jmi}-R^{HH} {}_{i \, \ j m}^{\ m}=2R^m_{jl} C^l_{mi}$.
Then, (\ref{kkg})  reads
\begin{equation}
\nabla^H_j (G_{i}^{\ j}+g^{jk}R^m_{ks} C^s_{im})=0,
\end{equation}
which implies\footnote{Under our assumption $I_j=0$ this contracted conservation is equivalent to (\ref{kkg}) contracted with $y^i$ because by (\ref{fyg}) and the symmetry of the Ricci tensor $(R^{HH}{}^m_{\ jmi}-R^{HH}{}_{i \ jm}^{\ m}) y^i=(R^{HH}{}^m_{\ jmi}-R^{HH}{}^m_{\ imj}) y^i=0$.}
\begin{equation}
\nabla_j^H (G^{j}_{\ i} y^i)=0 \quad \textrm{or equivalently}\quad \nabla_j^H (T^{j}_{\ i} y^i)=0 ,
\end{equation}
which can be considered the expression of  energy-momentum conservation for  every observer. Indeed, using a divergence theorem due to Rund \cite{rund75} for any section (observer) $s: M\to TM\backslash 0$, $g_{s}(s,s)=-1$, and field $A^k(x,y)\frac{\p}{\p y^k}$
\begin{equation}
\nabla^{s^* \!g} \cdot s^*\! A=s^*\!(\nabla^{HC} \cdot A)+s^*\!(I_j A^k+\frac{\p A^k}{\p y^j}) D_k s^j
\end{equation}
where $\nabla^{s^* \!g}$ is the Levi-Civita connection of the pullback metric $s^*\!g(x):=g(x,s(x))$. Thus
\begin{equation}
\nabla^{s^*\! g}_k (G^k_i(x,s(x)) s^i)= s^*\!(\frac{\p }{\p y^j} (G^k_i(x,y) y^i)) D_k s^j,
\end{equation}
 which is the analog of the general relativistic almost conservation equation $(T^{i j} u_j)_{; i}=T^{i j} u_{j; i}$.

To summarize, the dynamics is given by Eq.\ (\ref{lao}), (\ref{lop}) and (\ref{max}). The definition of $G_{ij}$ might need to be modified but the mechanism of conservation seems firmly established. In two previous works I have shown that the local causal structure of this theory is the correct one, with one past and one future cone \cite{minguzzi13c}, and that the global causal structure is similar to that of Lorentzian geometry \cite{minguzzi13d}.

One could have defined the stress-energy tensor by demanding  proportionality with  the  tensor appearing inside  parenthesis in Eq.\ (\ref{kkg}). However, with this choice the stress-energy tensor would not be symmetric, not even under our assumption $I_i=0$. This type of strategy is followed in \cite{li07} under the more restrictive Berwald case but with possibly $I_i\ne 0$ (please observe that in their  curvature tensors their second index plays the role of our first index).
\end{remark}

\subsection{Bianchi identities for the linear connection}

The connection $\nabla$ has lead us to a curvature $R^\nabla$, which for every pair of sections $\check X,\check Y \colon E \to TE$ provides an endomorphism $R^\nabla(\check X,\check Y)\colon VE \to VE$.

In order to write the Bianchi identity it is convenient to recall the following general approach to linear connections. Given a vector bundle $\pi_E\colon \tilde E \to E$, the covariant derivative is map
\[
\nabla\colon \Omega^0(E,\tilde E) \to \Omega^1(E,\tilde E), \quad  s \mapsto \nabla s
\]
which sends sections to  vector-valued 1-form fields. This map can be extended to a map
\[
d^\nabla\colon \Omega^p(E,\tilde E) \to \Omega^{p+1} (E,\tilde E),
\]
through the prescription
\[
d^\nabla(\Omega_p\otimes s)=d \Omega_p \otimes s+(-1)^p \Omega_p \wedge d^\nabla s
\]
where $\Omega_p$ is a $p$-form and $d^\nabla s:=\nabla s$. If $\eta$ is a vector valued form then
\[
d^\nabla (\Omega_p \wedge \eta)= d \Omega_p\wedge \eta+(-1)^p \Omega_p \wedge d^\nabla \eta.
\]
With our convention for the wedge product ($\alpha\wedge \beta =\alpha \otimes \beta -\beta \otimes \alpha$), the curvature $R^\nabla\in \Omega^2(E, \textrm{End} \tilde E)$ reads
\[
R^\nabla=d^\nabla \circ d^\nabla,
\]
that is, this expression coincides with that obtained from the usual formula which involves the covariant derivative, i.e.\ for $\check X,\check Y\in T_e E$
\[
R^\nabla(\check X, \check Y) s=\nabla_{\check X} \nabla_{\check Y} s-\nabla_{\check Y} \nabla_{\check X} s-\nabla_{[\check X,\check Y]} s.
\]
Similarly, given a soldering form $\sigma\in \Omega^1(E,\tilde E)$ the torsion is nothing but
\[
T^{(\nabla,\sigma)}=d^\nabla\sigma.
\]
It is possible to show that $d^\nabla d^\nabla \eta=R^\nabla \underset{\circ}{\wedge} \eta$, $\eta \in \Omega^1(E,\tilde E)$, where the small circle means that the endomorphism $R^\nabla$ is applied to the image vector of $\eta$.
The first Bianchi identities are
\begin{equation}
R^\nabla \underset{\circ}{\wedge} \sigma-d^\nabla T^{(\nabla,\sigma)}=0.
\end{equation}
Given a covariant derivative on the linear bundle $\tilde E\to E$ we have an induced covariant derivative on the  bundle $\textrm{End} \tilde E\to E$, given by $(\nabla_D \varphi)s=[\nabla_D, \varphi] s$, where $\varphi\in \Omega^0(E, \textrm{End}\tilde E)$ and $[,]$ is the  usual parenthesis of linear endomorphisms.
With this premise, the  second Bianchi identities are
\begin{equation}
d^\nabla R^\nabla=0,
\end{equation}
and follow easily by writing $(d^\nabla\circ d^\nabla)\circ d^\nabla s=d^\nabla\circ (d^\nabla\circ d^\nabla) s$ for an arbitrary section $s$.

In our case we have $\tilde{E}=VE$, and $E=TM\backslash 0$.  Thus for every $\check X,\check Y, \check Z\in T_e E$
\begin{align}
0&=\sum_{\textrm{cyclic}:\check X,\check Y, \check Z} \{R^\nabla(\check X, \check Y) \, \sigma(\check Z)-\nabla_{\check X} T^{(\nabla,\sigma)}(\check Y,\check Z)-T^{(\nabla,\sigma)}(\check X,[\check Y,\check Z])\}, \\
0&=\sum_{\textrm{cyclic}:\check X,\check Y, \check Z}  \{ [\nabla_{\check X}, R^\nabla(\check Y,\check Z)]+R^\nabla(\check X,[\check Y,\check Z]) \} . \label{esx}
\end{align}

The first Bianchi indentity leads to several more specialized identities. For space reasons we cannot explore all of them. We consider just the Berwald and Chern-Rund connections, in which case,  if we take all three vectors horizontal we obtain (\ref{llf}) in the horizontal soldering form case, and  (\ref{fpo})-(\ref{fpp}) in the vertical soldering form case. If we take one vector vertical and the other two vectors horizontal we obtain in the horizontal soldering form case
\begin{equation}
R^{VH}_{\textrm{Ber}} {}^l_{\,kij}=R^{VH}_{\textrm{Ber}} {}^l_{\,jik}, \qquad R^{VH}_{\textrm{ChR}} {}^l_{\,kij}=R^{VH}_{\textrm{ChR}} {}^l_{\,jik}
\end{equation}
which imply that the Berwald and Chern-Rund VH-Ricci tensors are symmetric (both equations are immediate from the expressions of these tensors). If we use the vertical soldering form we arrive at Eq.\ (\ref{msp}).

Let us study  the second Bianchi identity.

\subsubsection{HHH-second Bianchi identity} \label{xib}

Let us consider the case in which all vectors are horizonal $\check X=\mathcal{N}(X)$, $\check Y=\mathcal{N}(Y)$, $\check Z=\mathcal{N}(Z)$, where $X,Y,Z\colon M\to TM$.
Let $X=\p/\p x^i$, $Y=\p/\p x^j$, $Z=\p/\p x^k$, and let us apply Eq.\ (\ref{esx}) to $W=\p/\p y^s$.

The second  Bianchi identity becomes
\begin{equation} \label{mlx}
\sum_{\textrm{cyclic}:i,j,k}  \{ \nabla^H_i R^{HH}{}^l_{sjk}+ R^{VH}{}^l_{smi} R^m_{jk} +R^{HH}{}^l_{\, smi}(H^m_{kj}-H^m_{jk})\} =0 .
\end{equation}

%

This formula can be easily specialized to the notable Finsler connections using the previous formulas for the curvatures $R^{HH}$ and $R^{VH}$. For instance, in the Berwald case it becomes
\begin{equation} \label{sop}
\sum_{\textrm{cyclic}:i,j,k}  \{ \nabla^{HB}_i R^{HH}_{\textrm{Ber}}{}^l_{\, sjk}+ {\mathcal{G}}^l_{\,smi} R^m_{jk} \} =0 ,
\end{equation}
 for the Chern-Rund connection we have
\[
\sum_{\textrm{cyclic}:i,j,k}  \{ \nabla^{HC}_i R^{HH}_{\textrm{ChR}}{}^l_{\, sjk}+ (\frac{\p}{\p y^m} \,\Gamma^{l}_{si} ) R^m_{jk} \}  =0,
\]
 for the Cartan connection we have
\[
\sum_{\textrm{cyclic}:i,j,k}  \{ \nabla^{HC}_i R^{HH}_{\textrm{Car}}{}^l_{\, sjk}+ [\frac{\p}{\p y^m}\, \Gamma^{l}_{si} +C^l_{sn} L^n_{mi}-\nabla^{HC}_i C^l_{sm} ] R^m_{jk} \}  =0.
\]

\subsubsection{VHH-second Bianchi identity}

Let us consider the case in which $\check X$ is vertical and $\check Y$, $\check Z$, are horizontal. Let us set $\check X= \frac{\p}{\p y^i}$, $\check Y=\frac{\delta}{\delta x^j}$ and $\check Z= \frac{\delta}{\delta x^k}$ on Eq.\ (\ref{esx}).
After some algebras we arrive at the second Bianchi identity
\begin{align}
&\nabla^V_i R^{HH}{}^l_{sjk}+\nabla^{H}_k R^{VH} {}^l_{sij}-\nabla^H_j R^{VH} {}^l_{s ik}- R^{HH} {}^l_{skb} \,V^b_{ji}+R^{HH} {}^l_{sjb} \, V^b_{ki} \nonumber \\
&\ - R^{VV} {}^l_{sib} R^b_{jk}+ R^{VH}{}^l_{sbj} (H^b_{ik}-N^b_{ik})-R^{VH} {}^l_{sbk} (H^b_{ij}-N^b_{ij}) \nonumber\\
&\quad +R^{VH}{}^l_{sim} (H^m_{jk}-H^m_{kj})=0.
\end{align}
This expression simplifies considerably for the  Berwald connection
\begin{align}
&\frac{\p}{\p y^i} R^{HH}_{\textrm{ber}}{}^l_{sjk}+\nabla^{H}_k \mathcal{G}^l_{sij}-\nabla^H_j \mathcal{G}^l_{s ik}=0.
\end{align}
 For Chern-Rund it is given by
\begin{align}
&\frac{\p}{\p y^i} R^{HH}_{\textrm{ChR}}{}^l_{sjk}+\nabla^{HC}_k R^{VH}_{\textrm{ChR}} {}^l_{sij}-\nabla^{HC}_j R^{VH}_{\textrm{ChR}} {}^l_{s ik}- R^{VH}_{\textrm{ChR}} {}^l_{sbj}\, L^b_{ik}+R^{VH}_{\textrm{ChR}} {}^l_{sbk}\, L^b_{ij}=0,
\end{align}
while for Cartan it gives
\begin{align}
&\nabla^{VC}_i R^{HH}_{\textrm{Car}}{}^l_{sjk}+\nabla^{HC}_k R^{VH}_{\textrm{Car}} {}^l_{sij}-\nabla^{HC}_j R^{VH}_{\textrm{Car}} {}^l_{s ik}- R^{HH}_{\textrm{Car}} {}^l_{skb} \,C^b_{ji}+R^{HH}_{\textrm{Car}} {}^l_{sjb} \, C^b_{ki} \nonumber \\
& - R^{VV}_{\textrm{Car}} {}^l_{sib} R^b_{jk} - R^{VH}_{\textrm{Car}}{}^l_{sbj} \, L^b_{ik}+R^{VH}_{\textrm{Car}} {}^l_{sbk} \, L^b_{ij}=0.
\end{align}

\subsubsection{VVH-second Bianchi identity}

 Replacing $\check X=\frac{\p}{ \p y^i}$, $\check Y=\frac{\p}{ \p y^j}$, $\check{Z}=\frac{\delta}{\delta x^k}$ on (\ref{esx}) leads to the second Bianchi identity
\begin{align*}
0&=\nabla^V_i R^{VH} {}^l_{sjk}-\nabla^V_j R^{VH} {}^l_{sik}+\nabla^H_k R^{VV} {}^l_{sij}+R^{VH} {}^l_{sbk} (V^b_{ji}-V^b_{ij}) \nonumber \\
& \quad  +R^{VV} {}^l_{sib} (H^b_{jk}-N^b_{jk})-R^{VV} {}^l_{sjb} (H^b_{ik}-N^b_{ik})+R^{VH} {}^l_{sjb} V^b_{ki}-R^{VH}{}^l_{sib} V^b_{kj} .
\end{align*}
For the Berwald connection it reads
\[
\frac{\p}{\p y^i} \,\mathcal{G}^l_{sjk}=\frac{\p}{\p y^j}\, \mathcal{G}^l_{sik},
\]
which is also a consequence of the Schwarz equality of mixed partial derivatives. For the Chern-Rund connection it reads
\[
\frac{\p}{\p y^i}\, R^{VH}_{\textrm{ChR}}{}^l_{sjk}=\frac{\p}{\p y^j} \,R^{VH}_{\textrm{ChR}}{}^l_{sik}.
\]
For the Cartan connection it is given by the more complicated equality
\begin{align*}
0&=\nabla^{VC}_i R^{VH}_{\textrm{Car}} {}^l_{sjk}-\nabla^{VC}_j R^{VH}_{\textrm{Car}} {}^l_{sik}+\nabla^{HC}_k R^{VV}_{\textrm{Car}} {}^l_{sij} \nonumber \\
& \quad  -R^{VV}_{\textrm{Car}} {}^l_{sib} \,L^b_{jk}+R^{VV}_{\textrm{Car}} {}^l_{sjb} \, L^b_{ik}+R^{VH}_{\textrm{Car}} {}^l_{sjb} C^b_{ki}-R^{VH}_{\textrm{Car}} {}^l_{sib} C^b_{kj} .
\end{align*}

\subsubsection{VVV-second Bianchi identity}
Finally, by setting $\check X=\frac{\p}{ \p y^i}$, $\check Y=\frac{\p}{ \p y^j}$, $\check{Z}=\frac{\p}{ \p y^k}$ we arrive at the second Bianchi identity
\begin{align*}
0&=\nabla^V_i R^{VV} {}^l_{sjk}+\nabla^V_k R^{VV} {}^l_{sij}+\nabla^V_j R^{VV} {}^l_{ski}\nonumber \\
& \quad + R^{VV} {}^l_{sbk} (V^b_{ji}-V^b_{ij})+ R^{VV} {}^l_{sbj} (V^b_{ik}-V^b_{ki})+R^{VV} {}^l_{sbi} (V^b_{kj}-V^b_{jk}).
\end{align*}
This is the usual second Bianchi identity for $T_pM\backslash 0$ regarded as a pseudo-Riemannian manifold.
For the Berwald and Chern-Rund connections it gives a trivial equation since all  terms vanish. For the Cartan connection it reads
\begin{align*}
0&=\nabla^V_i R^{VV}_{\textrm{Car}} {}^l_{sjk}+\nabla^V_k R^{VV}_{\textrm{Car}} {}^l_{sij}+\nabla^V_j R^{VV}_{\textrm{Car}} {}^l_{ski},
\end{align*}
where $R^{VV}$ has the special form given by Eq.\ (\ref{kik}).

\section{Conclusion: what is the best Finsler connection?}

We have presented the theory of Finsler connections providing the main identities for the general case and the more interesting  special cases.  This study shows that the Berwald  Finsler connection is more elementary and more closely related to the spray and hence to the geodesics than the other connections. The Cartan connection being fully metric compatible  leads to a curvature with better symmetry properties, however, its  second Bianchi identities are more complicated and not particularly transparent. The Chern-Rund connection like the Berwald connection leads to nice first Bianchi identities, but it is less connected to geodesics than Berwald's.

We are  certainly tempted to  conclude that each Finsler connections has its own advantages and disadvantages. However, it seems more important to explain why it is difficult to select a `best' connection.

The geometry of the tangent space at a point may help to explain the difference between the Berwald and Chern-Rund connections on one side and the Hashiguchi and Cartan connections on the other side. The first two connections, by setting $V_{ij}^k=0$ regard $T_pM$ as a {\em teleparallel} space $\nabla_{\p/\p y^i} \frac{\p}{\p y^j}=0$, that is, they emphasize the fact that on $T_pM$ the directions  of any two vectors  $w_1\in T_{y_1} T_pM\backslash 0$, $w_2\in T_{y_2}T_pM\backslash 0$ can be compared even though $y_1\ne y_2$. Indeed, the coordinates  $\{ y^\mu\}$ are not arbitrary but are instead determined up to a linear transformation.

On the contrary in the Hashiguchi and Cartan approaches one emphasizes  the {\em pseudo-Riemannian} geometry of $(T_pM\backslash 0, g)$. Hence the connection is determined through the  compatibility with the metric. Here one regards $T_pM\backslash 0$ as an arbitrary manifold, thus the information on the special nature of the coordinate system $\{y^\mu\}$  is essentially lost.

This observation might explain why we cannot decide between these two approaches: each of them disregards some feature of the geometry at a point. In the former approach this aspect is the metric, in the latter approach it is the parallelism induced by the coordinate system $\{y^\mu\}$.

However, there is one aspect of the Berwald or Chern-Rund connections which is not particularly appealing. Namely, there is no curvature or torsion which measures the difference between a pseudo-Finsler space and a pseudo-Riemannian space.
For the Cartan or Hashiguchi connections we can use the torsion $T^{VH}_{\textrm{hor}}{}^i_{jk}=C^i_{jk}$.

Probably, the condition $V_{ij}^k=0$ implies that the Chern-Rund and Berwald connections retain too little information on the tangent space. We can try to remedy this situation as follows.
We impose that for some function $f>0$ the base $e_i=f \frac{\p}{\p y^i}$ is $\nabla^V$-parallely transported. This means that, as with Chern-Rund, the vertical connection preserves the teleparallelism induced by the coordinate system $\{y^\mu\}$ (although it does not preserve the `length' of the vectors $\p/\p y^i$). This condition reads
\[
0=\nabla_{e_i} e_j=f\nabla_i^V(f \frac{\p}{\p y^j})=f[\frac{\p f}{\p y^i}\,\delta_j^k+fV_{ji}^k] \frac{\p}{\p y^k}
\]
thus $V_{ji}^k=-\frac{\p \ln  f}{\p y^i} \, \delta^k_j$.
The  gained degree of freedom is used to solve $\nabla^V\mu=0$, that is, it is used to preserve the volume form. The study of Sect.\ \ref{vol} shows that this condition reads $I_s=V^r_{rs}$ thus we arrive at
\begin{equation}
V^k_{ji}=\frac{1}{n} \,I_i\, \delta^k_j.
\end{equation}
where $n$ is the dimension of $M$.
This Finsler connection is regular, that is,
\[
\det(\delta^k_i+\frac{1}{n}\,I_i \,y^k)\ne 0.
\]
Indeed, suppose that $v^i$ belongs to the kernel of $\delta^k_i+\frac{1}{n}I_i y^k$, then $v^k=-\frac{1}{n} (I_i v^i) y^k$, that is $v^k=a y^k$ for some $a$, but since $I_k y^k=0$ we have $a=0$, so the Finsler connection is regular.

Therefore, we are  led to the following Finsler connections
\[
(\mathcal{G}^i_j, \mathcal{G}^i_{jk}, \frac{1}{n}\, I_i\, \delta^k_j), \qquad (\mathcal{G}^i_j, \Gamma^i_{jk}, \frac{1}{n}\, I_i\, \delta^k_j).
\]
They improve Berwald and Chern-Rund, because
\[
T_{\textrm{hor}}^{VH} {}^k_{ij}=\frac{1}{n}\, I_i \delta^k_j, \qquad T_{\textrm{ver}}^{VV} {}^k_{ij}=\frac{1}{n} \,(I_i \delta^k_j-I_j \delta^k_i).
\]
In the positive definite case by Deicke's theorem
 $T_{\textrm{ver}}^{VV}=0$ if and only if the Finsler space is Riemannian. Thus the torsion  $T_{\textrm{ver}}^{VV}$ provides a measure of the difference between
Finsler and Riemannian spaces, thus solving the mentioned difficulty with the Berwald and Chern-Rund connections. Furthermore, contrary to what happens with the Cartan connection, the torsion that measures this difference, namely $T_{\textrm{ver}}^{VV}$, lives entirely in the vertical geometry, contrary to the tensor $T^{VH}_{\textrm{hor}}$ used for the same purpose in the Cartan or Hashiguchi approach, which instead, quite strangely, involves the horizontal directions. In other words, in the new connection we can appreciate the difference between Finsler and Riemannian spaces using vertical objects, while for the Cartan and Hashiguchi connections one has to probe the geometry of the space horizontally just to make `vertical'  conclusions.

By a similar line of argument the choice $H^i_{jk}=\Gamma^i_{jk}$ should be preferred over Berwald's in the definition of the horizontal covariant derivative. Indeed, using the Berwald or Hashiguchi connections there is not way to characterize the Landsberg spaces using their curvatures or torsions. They seem to contain too little information on the horizontal metric properties of the space. Instead, the choice $H^i_{jk}=\Gamma^i_{jk}$ allows one to recover the Berwald horizontal covariant derivative whenever the torsion $T^{VH}_{ver}$ vanishes. The other way around is precluded.

Thus, while it could be tempting to conclude that all the Finsler connections are similarly important and that each of them has its own merits,  the above geometrical considerations seem to point to a different direction. They  select the Finsler connections
\[
(\mathcal{G}^i_j, \Gamma^i_{jk}, C^i_{jk}), \qquad (\mathcal{G}^i_j, \Gamma^i_{jk}, \frac{1}{n}\, I_i\, \delta^k_j),
\]
as
those retaining most of the useful geometric information on the pseudo-Finsler space. It is possible that further analysis will reveal further improvements on the known Finsler connections. Ultimately, authors working in Finsler geometry  might indeed reach some consensus on a `best',  `right' or `most appropriate' Finsler connection.

\section*{Acknowledgments}
I thank Daniel Canarutto, Marco Modugno and Carlos Tejero Prieto for some very useful conversations on the theory of connections and Finsler geometry. I thank the Department of Mathematics of the University of Salamanca for kind hospitality. This work has been partially supported by GNFM of INDAM.
%


\begin{center}

\begin{tabular}{@{} *5l @{}}    \toprule
\emph{Math.\ objects} & \emph{Meaning} &  \\\midrule
$G^i$, $G^i_j$, $G^i_{jk}$, $G^i_{jkl}$    & General spray and its vertical derivatives \\
$\mathcal{G}^i$, $\mathcal{G}^i_j$, $\mathcal{G}^i_{jk}$, $\mathcal{G}^i_{jkl}$    & As above but $\mathcal{G}^i$ comes from the Finsler Lagrangian   \\
  $\mathcal{G}^i_{jkl}$    & Berwald's curvature\\
 $N^i_j$ & General non-linear connection \\
 $\mathcal{N}(X)$ & horizontal lift of $X$\\
 $\tilde X, X^V$ & vertical lift of $X$\\
  $\mathcal{G}^i_j$ & Berwald's non-linear connection \\
  $\mathcal{L}$    & Finsler Lagrangian  ($F^2/2$) \\
  $g_{ij}$    & Finsler metric \\
  $L=y^i{\p}/{\p y^i}$ & Liouville vector field \\
  $(N^i_j, H^i_{jk}, V^i_{jk})$ & General Finsler connection\\
  $(N^i_j, N^i_{jk}, 0)$ &  Finsler connection induced by a non-linear connection\\
    $(\mathcal{G}^i_j, \mathcal{G}^i_{jk}, 0)$ & Berwald's Finsler connection\\
  $(\mathcal{G}^i_j, \Gamma^i_{jk}, C^i_{jk})$ & Cartan's  Finsler connection\\
  $(\mathcal{G}^i_j, \Gamma^i_{jk}, 0)$ & Chern-Rund's Finsler connection\\
  $\nabla^{H}, \nabla^V$    & General horizontal and vertical covariant derivatives\\
   $\nabla^{HB}$    & Berwald's horizontal covariant derivative\\
      $\nabla^{HC}$    & Cartan's (Chern-Rund's) horizontal covariant derivative \\
      $\nabla^{VB}$    & Berwald's (Chern-Rund's) vertical covariant derivative\\
      $\nabla^{VC}$    & Cartan's vertical covariant derivative\\
   $C_{ijk}$    & Cartan torsion \\
    $I_i$    & Mean Cartan torsion \\
     $L^i_{jk}$    & Landsberg tensor \\
           $J_i$    & Mean Landsberg tensor \\
           $\mu$    & Volume form \\
              $R^m_{kl}$    & Curvature of the non-linear connection \\
              $\tau^m_{kl}$    & Torsion of the non-linear connection \\
                 $R^\nabla$    & Curvature of the (linear) Finsler connection \\
                    $R^{HH}$, $R^{VH}$, $R^{VV}$    & Projections of the curvature $R^\nabla$  \\
                    $\sigma$, $\sigma_h$, $\sigma_v$ & Soldering form, horizontal and vertical versions \\
                    $T^{(\nabla,\sigma)}$ & Torsion of the Finsler connection \\
                    $T^{HH}$, $T^{VH}$, $T^{VV}$    & Projections of the Torsion, in versions hor/ver  \\
           \bottomrule
 \hline
\end{tabular}

\end{center}

\end{document}